\documentclass[12pt, reqno]{amsart}

\makeatletter
\g@addto@macro{\endabstract}{\@setabstract}
\makeatother

\usepackage{graphics, stackrel}
\usepackage{amsmath, amssymb, amsthm}
\usepackage{graphicx}
\usepackage{verbatim}
\usepackage{amsfonts}
\usepackage{filecontents}
\usepackage{natbib}
\usepackage{fancyvrb}
\usepackage{color}

\usepackage[citecolor=blue, colorlinks=true, linkcolor=blue]{hyperref}

\usepackage{mathrsfs}
\usepackage{bbm}

\usepackage[left=1.25in, right=1.25in, top=1.0in, bottom=1.15in, includehead, includefoot]{geometry}

\usepackage{multirow}

\newcommand{\iidsim}{\stackrel {\textrm{ {\sc iid }}} {\sim} }
\newcommand{\1}{\mathbbm 1}

\newcommand*\diff{\mathop{}\!\mathrm{d}}
\renewcommand{\epsilon}{\varepsilon}

\newcommand{\cC}{\mathscr C}

\newcommand{\bB}{\mathscr B}

\newcommand{\hH}{\mathscr H}

\newcommand{\lL}{\mathscr L}

\newcommand{\fF}{\mathscr F}

\newcommand{\ZZ}{\mathsf Z}
\renewcommand{\SS}{\mathsf S}
\newcommand{\CC}{\mathsf C}
\newcommand{\DD}{\mathsf D}

\newcommand{\pP}{\mathscr P}

\newcommand{\RR}{\mathbbm R}

\newcommand{\NN}{\mathbbm N}
\newcommand{\PP}{\mathbbm P}

\newcommand{\EE}{\mathbbm E \,}

\theoremstyle{plain}

\newtheorem{theorem}{Theorem}[section]

\newtheorem{lemma}{Lemma}[section]
\newtheorem{proposition}{Proposition}[section]

\theoremstyle{definition}

\newtheorem{example}{Example}[section]

\newtheorem{assumption}{Assumption}[section]

\usepackage{caption}

\usepackage[flushleft]{threeparttable}

\newcommand{\me}{\mathrm{e}}

\usepackage{ dsfont }

\newcommand{\vertiii}[1]{{\left\vert\kern-0.25ex\left\vert\kern-0.25ex\left\vert #1 
    \right\vert\kern-0.25ex\right\vert\kern-0.25ex\right\vert}}
    
\usepackage{subcaption}
\usepackage{graphicx}  

\makeatletter
\newcommand*{\rom}[1]{\expandafter\@slowromancap\romannumeral #1@}
\makeatother
\begin{document}

\title{}

\date{\today}

\begin{center}
  \LARGE 
  The Income Fluctuation Problem with Capital Income Risk: 
  Optimality and Stability\footnote{We thank Jess Benhabib, Christopher Carroll, Fedor Iskhakov, Larry Liu, Ronald Stauber and Chung Tran for valuable feedback and suggestions, as well as audience members at the RSE  seminar at the Australian National University in 2017. \\
  \emph{Email addresses:} \texttt{qingyin.ma@anu.edu.au}, \, \texttt{john.stachurski@anu.edu.au}, \,
  \texttt{atoda@ucsd.edu}.
  }

  \bigskip
  \normalsize
  Qingyin Ma\textsuperscript{a}, \, John Stachurski\textsuperscript{b} \, and
  \, Alexis Akira Toda\textsuperscript{c} \par \bigskip

  \textsuperscript{a, b}Research School of Economics, The Australian National
  University \\
  \textsuperscript{c}Department of Economics, University of California, San Diego

  %\normalsize{\today}
  November 15, 2018
\end{center}

\begin{abstract} 
   This paper studies the income fluctuation problem with capital income risk (i.e., dispersion in the rate of return to wealth). Wealth returns and labor earnings are allowed to be serially correlated and mutually dependent. Rewards can be bounded or unbounded. Under rather general conditions, we develop a set of new results on the existence and uniqueness of solutions, stochastic stability of the model economy, as well as efficient computation of the ergodic wealth distribution. A variety of applications are discussed. Quantitative analysis shows that both stochastic volatility and mean persistence in wealth returns have  nontrivial impact on wealth inequality.
    \vspace{1em}
    
    \noindent
    %\textit{JEL Classifications:} XX, YY \\
\end{abstract}

%\maketitle

%section
\section{Introduction}

The income fluctuation problem refers to the broad class of decision problems that characterize the optimal consumption-saving behavior for agents facing stochastic income streams. In most cases, agents are subject to idiosyncratic shocks and borrowing constraints. Markets are incomplete so idiosyncratic risks cannot be fully diversified or hedged. The model represents one of the fundamental workhorses of modern macroeconomics, and has been adopted to study a large variety of important topics, ranging from asset pricing, life-cycle choice, fiscal policy, social security, to income and wealth inequality, among many others. See, for example, \cite{schechtman1976income}, 
\cite{deaton1992behaviour}, 
\cite{huggett1993risk}, \cite{aiyagari1994uninsured},
\cite{carroll1997buffer}, \cite{chamberlain2000optimal}, \cite{cagetti2008wealth},  \cite{de2010elderly}, \cite{guner2011taxation}, \cite{guvenen2011macroeconomics}, \cite{meghir2011earnings}, \cite{meyer2013consumption}, \cite{guvenen2014inferring} and \cite{heathcote2014consumption}.

%\footnote{See \cite{cagetti2008wealth} and \cite{guvenen2011macroeconomics} for overviews.}

In recent years, researchers have come to investigate an important mechanism in the income fluctuation framework---the dispersion in rates of return to wealth, referred to below as the capital income risk. Early studies are provided by \cite{angeletos2005incomplete} and \cite{angeletos2007uninsured}. These works highlight that the macroeconomic effects of idiosyncratic capital income risk can be both qualitatively distinct from those of idiosyncratic labor income risk and quantitatively significant. 

An especially important set of applications concerns wealth inequality. As is well known in the literature, the classic income fluctuation frameworks of \cite{huggett1993risk} and \cite{aiyagari1994uninsured}, in which returns to wealth are homogeneous across agents, fail to reproduce the high inequality and the fat upper tail of wealth distributions in many economies. Such empirical failure has prompted researchers to investigate models with uninsured capital income risk. Entrepreneurial risk, a representative example of capital income risk, is studied by \cite{quadrini2000entrepreneurship} and \cite{cagetti2006entrepreneurship}. By introducing heterogeneity across agents in their work and entrepreneurial ability, these studies successfully generate skewed wealth distributions that are more similar to those observed in the U.S. data.

Moreover, in an OLG economy with intergenerational transmission of wealth, \cite{benhabib2011distribution} show that capital income risk is the driving force of the heavy-tail properties of the stationary wealth distribution. In a Blanchard-Yaari style economy, \cite{benhabib2016distribution} show that idiosyncratic investment risk has a big impact on generating a double Pareto stationary wealth distribution. 
In another important contribution, \cite{gabaix2016dynamics} point out that a positive correlation of returns with wealth (``scale dependence") in addition to persistent heterogeneity in returns (``type dependence") can well explain the speed of changes in the tail inequality observed in the data.
An important work that is highly pertinent to the present paper is \cite{benhabib2015wealth}.
In a stylized infinite horizon income fluctuation problem with capital income risk, the authors prove that there exists a unique stationary wealth distribution that displays fat tail. 
 
On the empirical side, using twelve years of population data from Norway's administrative tax records, \cite{fagereng2016heterogeneity, fagereng2016heterogeneityb} document that individuals earn markedly different average returns to both their financial assets (a standard deviation of $14\%$) and net worth (a standard deviation of $8\%$). Wealth returns are heterogeneous both within and across asset classes. Returns are positively correlated with the wealth level and highly persistent over time. In addition, wealth returns are (mildly) correlated across generations.

Although theoretical, empirical and quantitative studies all reveal the significant economic impact of capital income risk, existing models of capital income risk in the income fluctuation framework are highly stylized. For example, the assumptions of {\sc iid} labor income process, {\sc iid} wealth return process and their mutual independence made by \cite{benhabib2015wealth} are rejected by the empirical data in several economies (see, e.g., \cite{kaplan2010much}, \cite{guvenen2010inferring} and \cite{fagereng2016heterogeneity, fagereng2016heterogeneityb}). As \cite{benhabib2015wealth} point out, adding positive correlations in labor earnings and wealth returns enriches model dynamics in that it captures economic environments with limited social mobility.

To our best knowledge, a general theory of capital income risk in the income fluctuation framework has been missing in the literature. This raises concerns about whether or not existing views on the economic impact of capital income risk hold in general, as well as whether or not modeling capital income risk in more generic and realistic settings is technically achievable. To be specific, several important questions are: 
\begin{itemize}
    \item Do correlations in the wealth return process (e.g., those caused by mean persistence or stochastic volatility of wealth returns) enhance or dampen the macroeconomic impact of capital income risk?
    
    \item What if, in addition to serial correlation, the wealth return process and the labor earnings process are mutually correlated?
    
    \item Does an optimal policy always exist in these generalized settings? If it does, is it unique?
    
    \item Does the stochastic law of motion for optimal wealth accumulation yield a stationary distribution of wealth?
    
    \item If it does, is the model economy globally stable, in the sense that the stationary distribution is unique and can be approached by the distributional path from any starting point?
    
    \item How do we compute the optimal policy and the stationary wealth distribution in practice?
\end{itemize}
These questions are highly significant, in the sense that a negative answer to any of them will pose a threat to the existing findings concerning capital income risk. However, due to technical limitations, these questions have not been investigated in a general income fluctuation framework. In this paper, we attempt to fill this gap. To this end, we extend the standard income fluctuation problem by characterizing the following essential features.
\begin{itemize}
    \item Agents face idiosyncratic rate of return to wealth $\{ R_t\}$ (capital 
    income risk) and idiosyncratic labor earnings $\{Y_t\}$ (labor income risk), both of which 
    are affected by a generic, exogenous Markov process $\{ z_t \}$.
    
    \item Supports of $\{ R_t\}$ and $\{ Y_t\}$ are bounded or unbounded, and, in 
    either case, allowed to contain zero.
    
    \item The reward (utility) function is  bounded or unbounded, and no specific structure is imposed beyond differentiability, concavity and the usual slope conditions. 
\end{itemize}
As can be seen, general $\{ R_t\}$ and $\{Y_t\}$ processes that are serially correlated and mutually dependent are covered by our framework. Moreover, consumption can become either arbitrarily small or arbitrarily large, so that agents are allowed to borrow up to the highest sustainable level of debt, creating rich and substantial model dynamics reflecting agents' borrowing activity.\footnote{See the discussion of \cite{rabault2002borrowing}.} 

%First, agents face idiosyncratic rate of return to wealth (capital income risk, denoted by $\{ R_t\}$) and labor earnings (labor income risk, denoted by  $\{Y_t\}$), both of which are affected by an exogenous Markov process $\{ z_t \}$. Hence, generic $\{ R_t\}$ and $\{Y_t\}$ processes with autocorrelations and mutual dependence are covered. Second, supports of $\{ R_t\}$ and $\{ Y_t\}$ are bounded or unbounded, and, in each case, allowed to contain zero. As a result, consumption can become either arbitrarily small or arbitrarily large, and agents are allowed to borrow up to the highest sustainable level of debt.\footnote{See, for example, \cite{rabault2002borrowing}.} Third, the reward (utility) function is  bounded or unbounded, and no specific structure is imposed beyond differentiability, concavity and the usual slope conditions. 

We make several tightly connected contributions on optimality, stochastic stability  and computation of this generalized income fluctuation problem. 

First, we prove that the Coleman operator adapted to this framework is indeed an ``$n$-step'' contraction mapping in a complete metric space of candidate consumption policies, even when rewards are unbounded. The unique fixed point is shown to be the optimal policy (also unique in the candidate space), and several important properties (e.g., continuity and monotonicity) are derived. To tackle unboundedness, we draw on and extend \cite{li2014solving} by adding capital income risk and constructing a metric that evaluates consumption differences in terms of marginal utility.  To obtain contractions under a minimal level of restriction, we focus our key assumption on bounding the \textit{long-run growth rate} of wealth returns. 

We show that this assumption is indeed equivalent to bounding the spectral radius of an expected wealth return operator (a bounded linear operator) by $1 / \beta$. As a result, it is similar to the assumptions made by recent literature regarding the operator theoretic method, which have been proven both necessary and sufficient for the existence and uniqueness of solutions in a variety of models (see, e.g., \cite{hansen2009long, hansen2012recursive}, \cite{borovivcka2017necessary} and \cite{toda2018wealth}). Our assumption is easy to verify numerically. For example, when the state space for the exogenous Markov process $\{z_t\}$ is finite, verifying this assumption is as convenient as finding the largest modulus of the set of eigenvalues for a given matrix.

Second, as our most significant contribution, we show that the model economy is globally stable, even in the presence of capital income risk. Specifically, there exists a unique stationary distribution for the state process (including wealth and the exogenous Markov state), and, given any initial state, the distributional path of the state process generated via optimal consumption and wealth accumulation converges to the stationary distribution as time iterates forward. The idea of proof goes as follows. Based on the optimality results established in the previous step, existence of a stationary distribution is guaranteed under some further restrictions on agents' level of patience, plus some mild assumptions on the stochastic properties of the exogenous state and the labor income processes. The key is to show that the wealth process is bounded in probability.

The proof of global stability is more tricky and separated into two scenarios. 

(Scenario \rom{1}) When the exogenous state process $\{z_t\}$ is independent and identically distributed, so are $\{R_t\}$ and $\{ Y_t\}$, and wealth is the only state variable remaining. We show that, with some additional concavity structure imposed, the model economy is monotone, allowing us to use some new results in the field of stochastic stability (due to \cite{kamihigashi2014stochastic, kamihigashi2016seeking}). Based on these results, both global stability and the Law of Large Numbers are established. In this case, convergence of the distributional path to its stationarity is in the form of weak convergence. Moreover, the added concavity assumption holds for standard utilities such as CRRA or the logarithm utility. Notably, even in the current case, our theory extends the stability theory of \cite{benhabib2015wealth}, since we allow $\{R_t\}$ and $\{ Y_t\}$ to be dependent on each other (a more detailed comparison is given below). 

(Scenario \rom{2}) When the exogenous state process $\{z_t\}$ is Markovian, $\{ R_t\}$ and $\{Y_t\}$ are in general autocorrelated and mutually dependent, and the structure of monotone economy is lost due to the added exogenous state. As a result, the order theoretic approach used in the previous case is no longer applicable. In response to that, we aim to exploit the traditional theory of stochastic stability (see, e.g., \cite{meyn2009markov}). Specifically, we provide sufficient conditions for the state process to be $\psi$-irreducible, strongly aperiodic and a positive Harris chain, which in turn guarantee global stability and the Law of Large Numbers. Convergence here is in total variation norm distance, which is stronger than weak convergence. Our sufficient conditions are easy to verify in applications, and centered around  existence of density representations for the exogenous state process and the labor earnings process. We only require that supports of the two densities contain respectively a nontrivial compact subset and a certain ``small" interval. Importantly, no further concavity structure is required.

Moreover, we show in this scenario that if we add the same concavity structure as we do in scenario \rom{1} and some other mild assumptions (e.g., existence of densities for the wealth return process and geometric drift property of the labor earnings process), then the model economy is indeed $V$-geometrically ergodic. As a result, convergence to the stationary distribution occurs at a geometric speed.  

Since an {\sc iid} process is a special Markov process, as a byproduct, the theory in scenario \rom{2} serves as an alternative stability theory when the exogenous state process is {\sc iid}. As can be seen from the discussion above, neither of the two theories is ``stronger" than the other in this circumstance. On the one hand, global stability in scenario \rom{1} is established under an additional concavity assumption, which is not required for global stability in scenario \rom{2}. On the other hand, we make no assumptions on the density structure of the key stochastic processes in scenario \rom{1} as we do in scenario \rom{2}.

Based on the established stability and ergodicity results, the unique stationary distribution can be approximated via tracking a single state process simulated according to the optimal consumption and wealth accumulation rules, which is highly efficient. The real caveat is that, in presence of capital income risk, there can be very large realized values of wealth (and  consumption), causing serious problems to numerical computation of the optimal policy.  However, this problem is alleviated in our setting. We show that, under our maintained assumptions, the optimal policy is concave and asymptotically linear with respect to the wealth level. Hence, at large levels of wealth, the optimal consumption rule can be well approximated via linear extrapolation. 

We provide several important applications. First, we illustrate how our theory can be applied to modeling capital income risk in different scenarios. Then, we provide a numerical example in which we explore the quantitative effect of stochastic volatility and mean persistence of the wealth return process on wealth inequality. In the calibrated economy, our quantitative analysis shows that both these two factors lead to lower tail exponents of the stationary wealth distribution and higher Gini coefficients, and thus a higher level of wealth inequality.

In terms of connections to the existing literature, the most closely related results are those found in the recent paper \cite{benhabib2015wealth}. Like us, the authors study capital income risk in an income fluctuation framework. On the one hand,  their paper proves an important theoretical result---the stationary wealth distribution has a fat tail, a topic not treated by the present paper (tail properties are only studied by us numerically). 

On the other hand, our theory of optimality and stochastic stability is considerably sharper and covers a much broader range of applications. Specifically, to avoid technical complication, \cite{benhabib2015wealth} assume that $\{R_t\}$ and $\{ Y_t\}$ are {\sc iid}, mutually independent, supported on bounded closed intervals with strictly positive lower bounds, and that their distributions are represented by densities. Albeit helpful for simplifying analysis and deriving tail properties, these assumptions rule out important features observed in the real economy (e.g., mean persistence and stochastic volatility in the empirical labor earnings and wealth return processes, as discussed). Moreover, the strictly positive lower bound for $\{ Y_t\}$ prevents agents from borrowing up to the highest sustainable level of debt, hiding substantial model dynamics.\footnote{As discussed in
    \cite{rabault2002borrowing}, in this case, agents are guaranteed a strictly 
    positive minimum level of consumption, so the marginal value of consumption is 
    bounded, and the problem can be easily solved by constructing 
    supremum norm contractions. However, relaxing this assumption allows agents 
    to systematically avoid exhausting their borrowing capacity. }
As described above, all these assumptions are relaxed in our framework.

Regarding earlier literature, specific types of capital income risk are modeled by \cite{quadrini2000entrepreneurship}, \cite{angeletos2005incomplete}, \cite{cagetti2006entrepreneurship} and \cite{angeletos2007uninsured} in general equilibrium frameworks. In comparison, the present paper focuses on constructing a ``more general" one-sector framework and deriving sharper theoretical results, which, of course, could potentially benefit ``more general" general equilibrium analysis.

Moreover, since we tackle unbounded rewards and the associated technical complication, our paper is also related to \cite{rabault2002borrowing}, \cite{carroll2004theoretical}, \cite{kuhn2013recursive} and \cite{li2014solving}. These works develop different methods to handle the issue of unboundedness in standard income fluctuation problems (ones without capital income risk). While \cite{carroll2004theoretical} constructs a weighted supremum norm contraction and works with the Bellman operator, the other three works focus on the Coleman operator. In particular, \cite{rabault2002borrowing} exploits the monotonicity structure, \cite{kuhn2013recursive} applies a version of the Tarski's fixed point theorem, while \cite{li2014solving} constructs a contraction mapping based on a metric that evaluates consumption differences in marginal values. As discussed above, the present paper draws on and extends \cite{li2014solving} by incorporating capital income risk. 

%A related paper is \cite{kuhn2013recursive}, in which the author treats unbounded rewards in a standard income fluctuation problem (without capital income risk). The idea is to construct a complete lattice of consumptions policies and apply Tarski's fixed point theorem, which in turn guarantees existence of an optimal policy. We conjecture that the method of Kuhn can be extended to treat the same optimization problem of the present paper. However, since it is hard to constract a complete lattice in which the candidates consumption policies are continuous with respect to all the state variables, such method could cause difficulty when deriving stochastic stability results. 

The rest of this paper is structured as follows. Section \ref{s:setup} formulates the problem. Section \ref{s:opt_results} establishes optimality results. Sufficient conditions for the existence and uniqueness of optimal policies are discussed. Section \ref{s:sto_stability} focuses on stochastic stability. Global stability and some further properties are studied. Section \ref{s:app} provides a set of applications. All proofs are deferred to the appendix.

\section{Set up}
\label{s:setup}

This section sets up the income fluctuation problem to be studied. As a first step, we introduce some mathematical techniques and notation used in this paper.

\subsection{Preliminaries}

\label{ss:prelm}

%If $\mu = \delta_x$, the Dirac measure concentrated on $x \in S$, we call 
%$\{ X_t \}_{t \geq 0}$ \emph{Markov-$(Q,x)$}. Moreover, we call 
%$\{ X_t \}_{t \geq 0}$ \emph{Markov-$Q$} if it is Markov-$(Q, \mu)$ for some 
%$\mu \in \pP(S)$.

%On the other hand, the \emph{right Markov operator} maps bounded measurable 
%function $h \colon S \to \RR$ into bounded measurable function $Qh$, where
%%
%\begin{equation*}
%  (Qh)(x) := \int h(y) Q(x, \diff y)   \qquad (x \in S).
%\end{equation*}
%%

% Given $\mu \in \pP(S)$ and stochastic kernel $Q$, an
%$S$-valued stochastic process $\{X_t \}_{t \geq 0}$ is called 
%\emph{Markov-$(Q,\mu)$} if $X_0$ has distribution $\mu$ and $Q(x, \cdot)$ is 
%the conditional distribution of $X_{t+1}$ given $X_t = x$.

Let $\NN$, $\RR$ and $\RR_+$ be the natural, real and nonnegative real numbers 
respectively. Given topological space $\SS$, let $\bB(\SS)$ be the Borel 
$\sigma$-algebra and let $\pP(\SS)$ be the set of probability measures on 
$\bB(\SS)$. A \emph{stochastic kernel} $Q$ on $\SS$ is a map 
$Q \colon \SS \times \bB(\SS) \to [0, 1]$ such that 
\begin{itemize}
    \item $x \mapsto Q(x, B)$ is $\bB(\SS)$-measurable for each $B \in \bB(\SS)$ 
    and 
    \item $B \mapsto Q(x, B)$ is a probability measure on $\bB(\SS)$ for each $x \in \SS$.  
\end{itemize}
Let $bc \SS$ be the set of bounded continuous functions on $\SS$. A stochastic kernel $Q$ is called \emph{Feller} if $x \mapsto \int h(y) Q(x, \diff y)$ is in $bc\SS$ whenever 
$h \in bc\SS$.

For all $t \in \NN$, we define the \emph{$t$-th order kernel} as
\begin{equation*}
	Q^1 := Q, \quad
	Q^t (x, B) := \int Q^{t-1}(y, B) Q (x, \diff y) \quad
	(x \in \SS, B \in \bB(\SS)).
\end{equation*} 
The value $Q^t (x,B)$ represents the probability of transitioning from 
$x$ to $B$ in $t$ steps.
Furthermore, for all $\mu \in \pP(\SS)$, we define $\mu Q^t \in \pP(\SS)$ as 
\begin{equation*}
  (\mu Q^t)(B) := \int Q^t (x, B) \mu(\diff x)  \qquad (B \in \bB(\SS)).
\end{equation*}
A sequence $\{ \mu_n \} \subset \pP(\SS)$ is called \emph{tight}, if, for all 
$\epsilon>0$, there exists a compact $K \subset \SS$ such that 
$\mu_n(\SS \backslash K) \leq \epsilon$ for all $n$. 
We say that \emph{$\mu_n$ converges to $\mu$ weakly} and write 
$\mu_n \stackrel{w}{\rightarrow} \mu$ if $\mu \in \pP(\SS)$ and
$\int h \diff \mu_n \rightarrow \int h \diff \mu$ for all bounded 
continuous $h \colon \SS \rightarrow \RR$.

A stochastic kernel $Q$ is called \emph{bounded in probability} if the
sequence $\{ Q^t(x, \cdot)\}_{t \geq 0}$ is tight for all $x \in \SS$.  We call
$\psi \in \pP(\SS)$ \emph{stationary} for $Q$ if $\psi Q = \psi$.  We say that
$Q$ is \emph{globally stable} if there exists a unique stationary distribution
$\psi$ in $\pP(\SS)$ and $\psi_0 Q^t \stackrel{w}{\to} \psi$ for all $\psi_0
\in \pP(\SS)$.

Let $K$ be a bounded linear operator from $bc \SS$ to itself and $\| \cdot \|$ be the supremum norm on $bc\SS$. The \textit{operator norm} and \textit{spectral radius} of $K$ are defined by
\begin{equation*}
    \| K \| := \sup \{ \|Kg\|: g \in bc\SS, \; \|g \| \leq 1 \}
    \quad \text{and} \quad
    r(K) := \lim_{m \to \infty} \|K^m \|^{1 / m}.
\end{equation*}
In particular, when $\SS$ is finite, $K$ becomes a square matrix, and the spectral radius $r(K)$ reduces to $\max_\lambda |\lambda|$, where $\lambda$ ranges over the set of eigenvalues of $K$. (See, e.g., page~663 of \cite{guide2006infinite}).

In what follows, $(\Omega, \fF, \PP)$ is a fixed probability space
on which all random variables are defined, while $\EE$ is 
expectations with respect to $\PP$.

\subsection{The income fluctuation problem}

We introduce capital income risk and consider a generalized income fluctuation problem as follows
\begin{align}
    \label{eq:trans_at}
    & \max \, \EE \left\{ 
                \sum_{t \geq 0} \beta^t u(c_t)
             \right\}    \nonumber \\
  \text{s.t.} \quad  
    & a_{t+1} = R_{t+1} (a_t - c_t) + Y_{t+1},  \\
    & 0 \leq   \; c_t \leq a_t, \quad (a_0, z_0)=(a,z) \text{ given} \nonumber,
\end{align}
where $\beta \in [0,1)$ is a state-independent discount factor, $u$ is the 
utility function, the control process $\{ c_t\}_{t \geq 0}$ is consumption,
$\{R_{t}\}_{t \geq 1}$ is a gross rate of return on wealth and 
$\{Y_{t} \}_{t \geq 1}$ is labor income. The return and income
processes obey
\begin{align}
\label{eq:RY_func}
  R_{t} &= R \left( z_{t}, \zeta_{t} \right),
      \quad  
      \left\{ \zeta_{t} \right\}_{t \geq 1} \iidsim \nu,  
      \nonumber  \\
  Y_{t} &= Y \left( z_{t}, \eta_{t} \right),
      \quad 
      \left\{ \eta_{t} \right\}_{t \geq 1} \iidsim \mu,
\end{align}
where $R$ and $Y$ are nonnegative real-valued measurable functions, $\{ \zeta_t\}$ and $\{ \eta_t\}$ are innovations, and
$\{z_t\}_{t \geq 0}$ is a time-homogeneous $\ZZ$-valued Markov process with 
Feller stochastic kernel $P$, where $\ZZ$ is a Borel subset of $\RR^m$ paired 
with the usual relative topology. 

Throughout we make the following assumption on the agent's utility.

\begin{assumption}
\label{a:utility}
    The utility function $u \colon \RR_+ \rightarrow \{ - \infty \} \cup \RR$
    is twice differentiable on $(0, \infty)$ and satisfies
    \begin{enumerate}
      \item $u' > 0$ and $u'' < 0$ everywhere on $(0, \infty)$, and
      \item $u'(c) \rightarrow \infty$ as $c \rightarrow 0$ and 
          $u'(c) \rightarrow 0$ as $c \rightarrow \infty$.
    \end{enumerate}
\end{assumption}

\begin{example}
A typical example that meets assumption \ref{a:utility} is the CRRA utility
\begin{equation}
\label{eq:crra_utils}
    u(c) = c^{1 - \gamma} / (1 - \gamma) 
    \quad \text{if } \gamma > 0, \, \gamma \neq 1
    \quad \text{and} \quad
    u(c) = \log c 
    \quad \text{if } \gamma = 1,
\end{equation}
where $\gamma > 0$ is the coefficient of relative risk aversion.
\end{example}

\subsection{Further notation}
\label{ss:nota}

We use $x$ and $\hat{x}$ to denote respectively the current and next period random variables.
In addition,
\begin{equation}
\label{eq:nota1}
  \EE_{a,z} 
      := \EE \left[ \,\cdot \, \big| \, (a_0,z_0) = (a, z) \right]
  \quad \text{and} \quad
  \EE_z 
      := \EE \left[ \, \cdot \, \big| \, z_0 = z \right].
\end{equation}
In particular, for any integrable function $f$, 
\begin{equation}
\label{eq:nota2}
    \EE_z \, f (\hat{z}, \hat{R}, \hat{Y})
  = \int f \left[
               \hat{z},
               R ( \hat{z}, \hat{\zeta} ), 
               Y \left( \hat{z}, \hat{\eta} \right)
            \right]
     P(z, \diff \hat{z})
     \nu(\diff \hat{\zeta})
     \mu(\diff \hat{\eta}).
\end{equation}

\section{Optimality Results}
\label{s:opt_results}

In this section, we show that, with bounded or unbounded rewards, the Coleman operator adapted to the income fluctuation problem above is an $n$-step contraction mapping on a complete metric space of candidate policies, and that the unique fixed point is the optimal policy. To that end, we make the following assumptions.

\begin{assumption}
	\label{a:ctra_coef}
	There exists $n \in \NN$ such that $\theta := \beta^n \sup_{z \in \ZZ} \EE_z R_1 \cdots R_n < 1$.
\end{assumption}

\begin{assumption}
    \label{a:Y_sum}
    For all $z \in \ZZ$, we have 
    $\sum_{t=1}^{\infty} \beta^t \EE_z Y_t < \infty$.
\end{assumption}

\begin{assumption}
    \label{a:bd_sup_ereuprm}
    $\sup_{z \in \ZZ} \EE_z \, \hat{R} < \infty$, \, $\sup_{z \in \ZZ} \EE_z u'(\hat{Y}) < \infty$ and $\sup_{z \in \ZZ} \EE_z \hat{R} u' (\hat{Y}) < \infty$.
\end{assumption}

\begin{assumption}
    \label{a:conti_ereuprm}
    The functions 
    $z \mapsto R(z, \zeta)$, $z \mapsto Y(z, \eta)$,
    $z \mapsto \EE_z \hat{R}$ and 
    $z \mapsto \EE_z \hat{R} \, u'(\hat{Y})$ are continuous.
\end{assumption}

\begin{example}
	\label{ex:spec_rad_ctra}
	For all bounded continuous function $f$ on $\ZZ$, define
	\begin{equation*}
	    K f (z) := \EE_z \hat{R} f(\hat{z}),
	    \quad z \in \ZZ.
	\end{equation*}
	Then $K$ is a bounded linear operator by assumption \ref{a:bd_sup_ereuprm}. 
	Let $r(K)$ be the spectral radius of $K$. Then assumption \ref{a:ctra_coef} 
	holds if and only if $\beta r(K) < 1$. We prove this result in the appendix. 
\end{example}

\begin{example}
	\label{ex:spec_rad_matrix}
	Let $\{z_t\}$ be a finite-state Markov chain on 
	$\ZZ := \{ i_1, \cdots, i_N \}$ with transition matrix $\Pi$ (a ``discrete'' 
	stochastic kernel). Let $\text{diag} (\cdot)$ denote the diagonal matrix 
	generated by elements in the bracket, and, with slight abuse of notation, let
	\begin{equation*}
	    \EE R(z, \zeta) := \int R(z, \zeta) \nu (\diff \zeta)
	    \quad \text{and} \quad
	    D := \text{diag} \left( \EE R(i_1, \zeta), \cdots , \EE R(i_N, \zeta) \right).       
	\end{equation*}
	In this case, the operator $K$ in example \ref{ex:spec_rad_ctra} reduces to the matrix $K = \Pi D$. Therefore, assumption \ref{a:ctra_coef} holds if and only if $r(\Pi D) < 1 / \beta$. In particular, $r(\Pi D)$ equals the largest modulus of all the eigenvalues of $\Pi D$.
\end{example}

\begin{example}
	\label{ex:CS_suff}
	Based on the H{\"o}lder's inequality, to show assumption  \ref{a:bd_sup_ereuprm}, 
	it suffices to find some $p, q \in [1, \infty]$ such that $1/p + 1/q = 1$ and 
	\begin{equation*}
		\sup_{z \in \ZZ} \EE_z \, \hat{R}^p  < \infty
		\quad \text{and} \quad
		\sup_{z \in \ZZ} \EE_z u' ( \hat{Y} )^q < \infty.
	\end{equation*}
\end{example}

To establish the required results, we (temporarily) assume $a_0>0$ 
and set the asset space as $(0, \infty)$. The state space for the state process
$\{(a_t, z_t) \}_{t \geq 0}$ is then\footnote{Note that the second condition of 
    assumption \ref{a:utility} and assumption \ref{a:bd_sup_ereuprm} imply that 
    $\PP \{ Y_t > 0 \}=1$ for all $t \geq 1$ (although $Y_t$ is allowed to be 
    arbitrarilly close to zero). Hence, $\PP \{ a_t > 0 \} = 1$ for all $t \geq 1$ 
    by the law of motion \eqref{eq:trans_at}. It thus makes no difference 
    to optimality to exclude zero from the asset space. Doing 
    this simplifies analysis since $u$ and $u'$ are finite away from 
    zero. It actually allows us to propose a useful metric and apply the 
    contraction approach, as to be shown later.}
\begin{equation*}
\label{eq:S0}
  \SS_0:= (0, \infty) \times \ZZ \ni (a,z).
\end{equation*}
Consider the maximal asset path $\{ \tilde{a}_t \}$ defined by 
\begin{equation}
\label{eq:max_path}
    \tilde{a}_{t+1} = R_{t+1} \, \tilde{a}_t + Y_{t+1}
    \quad \text{and}
    \quad (\tilde{a}_0, \tilde{z}_0) = (a,z) \; \text{given}.
\end{equation}

\begin{lemma}
\label{lm:max_path}
If assumptions \ref{a:ctra_coef}--\ref{a:Y_sum} hold,
then $\sum_{t \geq 0} \beta^t \EE_{a,z} \, \tilde{a}_t$ is finite 
for all $(a,z) \in \SS_0$. 
\end{lemma}

A \emph{feasible policy} is a Borel measurable function 
$c \colon \SS_0 \rightarrow \RR$ with $0 \leq c(a,z) \leq a$ for all 
$(a,z) \in \SS_0$. Given any feasible policy $c$ and initial condition 
$(a,z) \in \SS_0$, the \emph{asset path} generated by $(c, (a,z))$ is the 
sequence $\{ a_t\}_{t \geq 0}$ in \eqref{eq:trans_at} when 
$c_t = c (a_t, z_t)$ and $(a_0, z_0) = (a,z)$. The \emph{lifetime value} of 
any feasible policy $c$ is the function 
$V_c \colon \SS_0 \rightarrow \{ - \infty \} \cup \RR$ defined by 
\begin{equation*}
  V_c (a,z) = \EE_{a,z} 
               \left\{ 
                  \sum_{t \geq 0} \beta^t u \left[ c (a_t, z_t) \right]
               \right\}, 
\end{equation*}
where $\{ a_t\}$ is the asset 
path generated by $(c,(a,z))$. Notice that $V_c(a,z) < \infty$ for any 
feasible $c$ and any $(a,z) \in \SS_0$. This is because, by assumption 
\ref{a:utility}, there exists a constant $L$ such that $u(c) \leq c + L$, 
and hence 
\begin{equation*}
  V_c(a,z) \leq \EE_{a,z} \sum_{t \geq 0} \beta^t u(a_t)
           \leq \EE_{a,z} \sum_{t \geq 0} \beta^t u(\tilde{a}_t)
           \leq \sum_{t \geq 0} \beta^t \EE_{a,z} \, \tilde{a}_t
                + \frac{L}{1 - \beta}.
\end{equation*}
The last expression is finite by lemma \ref{lm:max_path}.

A feasible policy $c^*$ is called \emph{optimal} if $V_c \leq V_{c^*}$ on
$\SS_0$ for any feasible policy $c$. In the present setting, the finiteness of 
$V_c$ for each feasible policy, the strict concavity of $u$, and the 
convexity of the set of feasible policies from each $(a,z) \in \SS_0$ imply 
that for each given parameterization, at most one optimal policy exists.

A feasible policy is said to satisfy the \textit{first order optimality 
conditions} if 
\begin{equation}
\left( u'\circ c \right)(a,z) \geq 
   \beta \, \EE_{z} \, \hat{R} 
              \left( u' \circ c \right)
              \left( 
                  \hat{R} \left[ a - c(a,z) \right] + \hat{Y}, 
                  \, \hat{z} 
              \right)
\end{equation}
for all $(a,z) \in \SS_0$, and equality holds when 
$c(a,z) < a$. Moreover, a feasible policy is said to satisfy the 
\textit{transversality condition} if, for all $(a, z) \in \SS_0$,
\begin{equation}
\label{eq:tvc}
  \lim_{t \rightarrow \infty}
    \beta^t       
    \EE_{a,z} 
        \left[ 
            \left(u' \circ c \right)(a_t, z_t) \, a_t  
        \right] = 0.
\end{equation}

\begin{theorem}
\label{t:opt_result}
If assumptions \ref{a:utility} and \ref{a:ctra_coef}--\ref{a:Y_sum} hold, and $c$ is a feasible policy that satisfies both the first order optimality 
conditions and the transversality condition, then $c$ is an optimal policy.
\end{theorem}

When does an optimal policy exist, and how can we compute it? To answer these 
questions, following \cite{li2014solving}, we use a contraction argument, where the 
underlying function space is set to $\cC$, the functions 
$c \colon \SS_0 \rightarrow \RR$ such that
\begin{enumerate}
  \item $c$ is continuous,
  \item $c$ is increasing in the first argument,
  \item $0 < c(a,z) \leq a$ for all $(a,z) \in \SS_0$, and
  \item $\sup_{(a,z) \in \SS_0} \left| (u' \circ c)(a,z) - u'(a) \right| < \infty$.
\end{enumerate}
To compare two policies, we pair $\cC$ with the distance
\begin{equation}
\label{eq:rho_metric}
  \rho(c,d) 
    := \left\| u' \circ c - u' \circ d \right\|
    := \sup_{(a,z) \in \SS_0} 
             \left| 
                 \left(u' \circ c \right)(a,z) - 
                 \left(u' \circ d \right)(a,z) 
             \right|
\end{equation}
that evaluates the maximal difference in terms of marginal utility.
Note that 
\begin{equation}
\label{eq:bd_uprime}
  c \in \cC \Longrightarrow 
    \exists K \in \RR_+ 
    \text{ s.t. } 
    u'(a) \leq (u' \circ c)(a,z) \leq u'(a) + K, \,
    \forall (a,z) \in \SS_0.
\end{equation}
Moreover, while elements of $\cC$ are not generally bounded, one can show 
that $\rho$ is a valid metric on $\cC$. In particular, $\rho$ is finite on 
$\cC$ since 
$\rho(c,d) \leq \left\| u' \circ c - u' \right\| 
              + \left\| u' \circ d - u' \right\|$, 
and the last two terms are finite by the definition of $\cC$.

\begin{proposition}
\label{pr:complete}
    $(\cC, \rho)$ is a complete metric space.
\end{proposition}

\begin{proposition}
\label{pr:suff_optpol}
If assumptions \ref{a:utility} and \ref{a:ctra_coef}--\ref{a:bd_sup_ereuprm} 
hold, $c \in \cC$, and, for all $(a,z) \in \SS_0$, 
\begin{equation}
\label{eq:foc}
  \left( u' \circ c \right)(a,z) = 
    \max \left\{ 
            \beta \, \EE_z \, \hat{R} 
                      \left( u' \circ c \right) 
                        \left( \hat{R} \left[ a - c(a,z)\right] + \hat{Y}, 
                               \, \hat{z} 
                        \right),
            u'(a)
         \right\},
\end{equation}
then $c$ satisfies both the first order optimality conditions and the 
transversality condition. In particular, $c$ is an optimal policy.
\end{proposition}

Inspired by proposition \ref{pr:suff_optpol}, we aim to characterize the 
optimal policy as the fixed point of the \emph{Coleman operator} $T$ defined 
as follows: for fixed $c \in \cC$ and $(a,z) \in \SS_0$, the value of the 
image $Tc$ at $(a,z)$ is defined as the $\xi \in (0,a]$ that solves
\begin{equation}
\label{eq:T_opr}
    u'(\xi) = \psi_c(\xi, a, z),
\end{equation}
where $\psi_c$ is the function on 
\begin{equation}
\label{eq:dom_T_opr}
    G := \left\{ 
            (\xi, a, z) \in \RR_+ \times (0, \infty) \times \ZZ \colon
            0 < \xi \leq a
         \right\}
\end{equation}
defined by
\begin{equation}
\label{eq:keypart_T_opr}
    \psi_c(\xi,a,z) := 
      \max \left\{
              \beta \EE_{z} \hat{R}
                 (u' \circ c)[\hat{R}(a - \xi) + \hat{Y}, \, \hat{z}], 
              \, u'(a)
           \right\}.
\end{equation}

The following propositions show that the Coleman operator $T$ is a 
well-defined self-map from the candidate space $(\cC, \rho)$ into itself.

\begin{proposition}
\label{pr:welldef_T}
If assumptions \ref{a:utility} and \ref{a:ctra_coef}--\ref{a:bd_sup_ereuprm} 
hold, then
  for each $c \in \cC$ and $(a,z) \in \SS_0$,
  there exists a unique $\xi \in (0,a]$ that solves \eqref{eq:T_opr}.
\end{proposition}

\begin{proposition}
\label{pr:self_map}
  If assumptions \ref{a:utility} and \ref{a:ctra_coef}--\ref{a:conti_ereuprm} hold, then $Tc \in \cC$ for all $c \in \cC$.
\end{proposition}

Recall $n$ and $\theta$ defined in assumption \ref{a:ctra_coef}. We now provide our key optimality result.

\begin{theorem}
\label{t:ctra_T}
If assumptions \ref{a:utility} and \ref{a:ctra_coef}--\ref{a:conti_ereuprm} hold, 
then $T^n$ is a contraction mapping on $(\cC, \rho)$ with modulus $\theta$. 
In particular,
\begin{enumerate}
    \item $T$ has a unique fixed point $c^* \in \cC$.
    \item The fixed point $c^*$ is the unique optimal policy in $\cC$.
    \item For all $c \in \cC$ and $k \in \NN$, we have 
      $\rho(T^{nk} c, c^*) \leq \theta^k \rho(c, c^*)$.
\end{enumerate}
\end{theorem}

\section{Stochastic Stability}
\label{s:sto_stability}

This section focuses on stochastic stability of the generalized income fluctuation
problem. We first provide sufficient conditions for the existence of a
stationary distribution and then explore conditions for uniqueness and
ergodicity.

Now we add zero back into the asset space, and consider a larger state space for the state process $\{ (a_t, z_t)\}_{t \geq 0}$, denoted by
\begin{equation*}
\label{eq:S}
  \SS := [0, \infty) \times \ZZ \ni (a,z).
\end{equation*}
We extend $c^*$ to $\SS$  by setting $c^* (0,z) = 0 $ for all $z \in \ZZ$.
Together, $c^*$ and the transition functions for $\{ a_t\}$, 
$\{ R_t\}$ and $\{ Y_t\}$ determine a Markov process with state vector 
$s_t := (a_t, z_t)$ taking values in the state space $\SS$. 
Let $Q$ denote the corresponding stochastic kernel.
The law of motion of $\{ s_t\}$ is
\begin{align}
\label{eq:dyn_sys}
  a_{t+1} &= R \left( 
                  z_{t+1}, \zeta_{t+1} 
               \right) 
             \left[ a_t - c^* \left(a_t, z_t \right) \right]
             + Y \left( 
                    z_{t+1}, \eta_{t+1} 
                 \right), 
                 \nonumber  \\
  %\tag{DS1}
  z_{t+1} &\sim P \left( z_t, \, \cdot \, \right)
\end{align}

\subsection{Existence of a stationary distribution}
\label{ss:exist_stat}

To obtain existence of a stationary distribution, we make the following assumptions.

\begin{assumption}
    \label{a:suff_bd_in_prob}
    There exists $\alpha \in (0,1)$ such that
    \begin{enumerate}
      \item  $\beta \, \EE_z \hat{R} \, u'[ \hat{R} \left(1-\alpha \right)a ]  \leq u'(a)$
      for all $(a,z) \in \SS_0$,\footnote{Here we adopt the convention that 
          $0 \cdot \infty = 0$ so that assumption \ref{a:suff_bd_in_prob} does not 
          rule out the case $\PP \{R_t =0 \mid z_{t-1} = z\} > 0$. Indeed, as would be 
          shown in proofs, all the conclusions of 
          this paper still hold if we replace this condition 
          by the weaker alternative: 
          $\beta \, \EE_z \hat{R} \, u'[ \hat{R} \left(1-\alpha \right)a + \alpha \hat{Y}]  
            \leq u'(a)$
          for all $(a,z) \in \SS_0$, while maintaining the second part of  
          assumption \ref{a:suff_bd_in_prob}.}
      and   
      \item there exists $n \in \NN$ such that 
      $(1 - \alpha)^n \sup_{z \in \ZZ} \EE_z R_1 \cdots R_n < 1$.
    \end{enumerate}
\end{assumption}

\begin{assumption}
    \label{a:bd_in_prob_Yt}
    $\sup_{t \geq 1} \EE_z \, Y_t < \infty$ for all $z \in \ZZ$.
\end{assumption}

\begin{assumption}
    \label{a:z_bdd_in_prob}
    The stochastic kernel $P$ is bounded in probability.
\end{assumption}

\begin{example}
\label{ex:homog}
    For homogeneous utility functions (e.g., CRRA), if the first condition of assumption \ref{a:suff_bd_in_prob} holds for some $a \in (0, \infty)$, then it must hold for all $a \in (0, \infty)$. To see this, let $k$ be the degree of homogeneity. Then we have
    \begin{equation*}
        \beta \EE_z \hat{R} u'[\hat{R} (1 - \alpha) a] / u'(a)
        = \beta \EE_z \hat{R}^{1+k} (1 - \alpha)^k 
        \quad \text{for all } a \in (0, \infty).
    \end{equation*}
    The right hand side is constant in $a$.
\end{example}

\begin{example}
\label{ex:bdd_in_prob_matrix}
Recall example \ref{ex:spec_rad_matrix}, where $\{ z_t \}$ is a finite-state Markov chain. Consider the CRRA utility defined in \eqref{eq:crra_utils}. Define further the column vector 
\begin{equation*}
    V := \left(
            \EE R(i_1, \zeta)^{1 - \gamma}, \cdots, \EE R(i_N, \zeta)^{1 - \gamma} 
        \right)'.
\end{equation*}
Then, assumption \ref{a:suff_bd_in_prob} holds whenever 
\begin{equation}
\label{eq:bdd_in_prob_matrix}
    \max \{ r(\Pi D), 1 \} < \left( \beta \| \Pi V \| \right)^{-1 / \gamma}.
\end{equation}
To see this, the first condition of assumption \ref{a:suff_bd_in_prob} holds if there exists $\alpha \in (0,1)$ such that $(1 - \alpha)^{-\gamma} \beta \EE_z \hat{R}^{1 - \gamma} \leq 1$ for all $z \in \ZZ$. Since $\ZZ$ is finite, this is equivalent to the existence of an $\alpha \in (0,1)$ such that $(1 - \alpha)^{-\gamma} \beta \| \Pi V \| \leq 1$. Similar to example \ref{ex:spec_rad_matrix}, the second condition of assumption \ref{a:suff_bd_in_prob} holds if $r(\Pi D) < 1 / (1 - \alpha)$ for the same $\alpha$. Together, these requirements are equivalent to \eqref{eq:bdd_in_prob_matrix}.
\end{example}

\begin{example}
    \cite{benhabib2015wealth} consider the CRRA utility and assume that 
    $\left\{ R_{t} \right\}$ and $\left\{ Y_{t} \right\}$ are {\sc iid},
    mutually independent, supported on bounded closed intervals of strictly positive 
    real numbers with their distributions represented by densities, and that
    $\beta \EE R_t^{1 - \gamma} < 1$ and
    $( \beta \EE R_t^{1 - \gamma} )^{\frac{1}{\gamma}} \EE R_t < 1$.
    Under these conditions, 
    assumptions \ref{a:bd_in_prob_Yt}--\ref{a:z_bdd_in_prob} obviously hold.
    Assumption \ref{a:suff_bd_in_prob} is satisfied by letting  
    $\alpha := 1 - ( \beta \EE R_t^{1 - \gamma} )^{\frac{1}{\gamma}}$
    and $n := 1$. The first condition of assumption 
    \ref{a:suff_bd_in_prob} holds since $\alpha \in (0,1)$ and
    \begin{align*}
      \beta \EE_z \hat{R} \, u' [ \hat{R} (1 - \alpha) a ]
           \big/ u'(a) 
       = (1 - \alpha)^{-\gamma} \beta \EE R_t^{1 - \gamma}
       = \left( \beta \EE R_t^{1 - \gamma} \right)^{-1} \beta \EE R_t^{1-\gamma} 
       = 1,
    \end{align*}
    while the second condition holds for $n=1$ since 
    $(1 - \alpha) \EE R_t 
        = ( \beta \EE R_t^{1 - \gamma} )^{\frac{1}{\gamma}} \EE R_t
        < 1.$
\end{example}

Let $c^*$ be the unique optimal policy obtained from theorem \ref{t:ctra_T} and 
$\alpha$ be defined as in assumption \ref{a:suff_bd_in_prob}. 
The next proposition establishes a strictly positive lower bound on the optimal
consumption rate.

\begin{proposition}
    \label{pr:opt_pol_bd_frac}
    If assumptions \ref{a:utility}, \ref{a:ctra_coef}--\ref{a:conti_ereuprm} and
    \ref{a:suff_bd_in_prob} hold, then
    $c^*(a,z) \geq \alpha a$ for all $(a,z) \in \SS$.
\end{proposition}

From this result the existence of a stationary distribution is not difficult
to verify.

\begin{theorem}
    \label{t:sta_exist}
    If assumptions \ref{a:utility}, \ref{a:ctra_coef}--\ref{a:conti_ereuprm} and 
    \ref{a:suff_bd_in_prob}--\ref{a:z_bdd_in_prob} hold, then
    $Q$ is bounded in probability and admits at least one stationary 
    distribution.
\end{theorem}

\subsection{Further Optimality Properties}
\label{ss:further_prop}

Slightly digressed from our main topics, we show that the optimal policy satisifies several other important properties under the following assumption.

\begin{assumption}
    \label{a:concave}
    The map
    $ s \mapsto
          (u')^{-1} \left[ 
                       \beta \EE_z \hat{R}
                           \left( u' \circ c \right) 
                           (\hat{R} s + \hat{Y}, \, \hat{z} )
                    \right]$
    is concave on $\RR_+$ for each fixed $z \in \ZZ$ and $c \in \cC$ that is 
    concave in its first argument.
\end{assumption}

\begin{example}
\label{eg:concave}
Assumption \ref{a:concave} imposes some concavity structure on the utility 
function. It holds for CRRA and logarithmic utilities, as shown in appendix B.
\end{example}

The next proposition implies that, with this added concavity structure, the optimal 
policy is concave and asymptotically linear with respect to the wealth level. 

\begin{proposition}
    \label{pr:optpol_concave}
    If assumptions \ref{a:utility}, \ref{a:ctra_coef}--\ref{a:conti_ereuprm},
    \ref{a:suff_bd_in_prob} and \ref{a:concave} hold, then
    
    \begin{enumerate}
        \item $a \mapsto c^*(a,z)$ is concave for all $z \in \ZZ$, and
                
        \item for all $z \in \ZZ$, there exists $\alpha' \in [\alpha,1)$ such that
        $\lim_{a \to \infty} [c^*(a,z) / a] = \alpha'$.\footnote{\label{fn:abar<inf}	
             Here we rule out the trivial situation 
             $\PP \{R_t = 0 \mid z_{t-1} = z \} = 1$, in which case
             $\alpha' = 1$.}
    \end{enumerate}
\end{proposition}

By proposition \ref{pr:optpol_concave}, as $a$ gets large, $c^*(a, z) \approx \alpha' a + b(z)$ for some function $b$, which is helpful for numerical computation. In the presence of capital income risk, there can be large realized values of wealth and consumption. This proposition then provides a justification for the linear extrapolation technology adopted when computing the optimal policy at large wealth levels.

\subsection{Global stability}
\label{ss:glb_stb}

We start with the case of {\sc iid} $\{z_t \}$ process, which allows us 
to exploit the monotonicity structure of the stochastic kernel $Q$. 
We then discuss general Markov $\{ z_t\}$ processes.
Since $Q$ is not generally monotone in these  
settings,\footnote{Since the optimal 
    policy $c^*(a,z)$ is not generally monotone in $z$, we cannot 
    conclude from \eqref{eq:dyn_sys} that $a_{t+1}$ is monotone in 
    $z_t$. Hence, $(a_{t+1}, z_{t+1})$ is not necessarily increasing 
    in $(a_t, z_t)$ and monotonicity might fail.}
global stability is established via a different approach.

\subsubsection{Case I: {\sc iid} $\{z_t\}_{t \geq 0}$ process}
\label{ss:gs_iid}

In this case, both $\{ R_t\}$ and $\{ Y_t\}$ are {\sc iid} processes, 
though dependence between $\{ R_t\}$ and $\{ Y_t\}$ are allowed. 
The optimal policy is then a function of asset only, and the transition 
function \eqref{eq:dyn_sys} reduces to
\begin{equation}
\label{eq:dyn_sys_iid}
  a_{t+1} 
  = R_{t+1} \left[ a_t - c^*(a_t) \right] + Y_{t+1}.
\end{equation}
In particular, we have a Markov process $\{ a_t \}_{t \geq 0}$ taking values in $\RR_+$. The next result extends theorem~3 of \cite{benhabib2015wealth}.

\begin{theorem}
    \label{t:gs_iid}
    If assumptions \ref{a:utility}, \ref{a:ctra_coef}--\ref{a:conti_ereuprm}, 
    \ref{a:suff_bd_in_prob}--\ref{a:bd_in_prob_Yt} and \ref{a:concave} hold, then $Q$ is globally stable.\footnote{Since 
          $\{ z_t\}$ is {\sc iid}, conditional expectations reduce to unconditional ones. 
          Hence, to verify assumptions \ref{a:ctra_coef}--\ref{a:conti_ereuprm} and 
          \ref{a:bd_in_prob_Yt}, it suffices to show:
          $\EE R_t^2 < \infty$, $\beta \EE R_t < 1$,
          $\EE Y_t < \infty$ and $\EE [u'(Y_t)]^2 < \infty$.} 
\end{theorem}

Let $\psi^*$ be the unique stationary distribution of $Q$, obtained in 
theorem \ref{t:gs_iid}. Let $\lL$ be the linear span of the set of 
increasing $\psi^*$-integrable functions
$h \colon \RR_+ \to \RR$.\footnote{In other words, 
    $\lL$ is the set of all $h \colon \RR_+ \to \RR$ such that 
    $h = \alpha_1 h_1 + \cdots + \alpha_k h_k$ for some scalars
    $\{ \alpha_i\}_{i=1}^k$ and increasing measurable $\{h_i \}_{i=1}^k$
    with $\int |h_i| \diff \psi^* < \infty$.} 
Recall that $bc \RR_+$ is the set of continuous bounded functions 
$h \colon \RR_+ \to \RR$. The following theorem shows that the Law of Large 
Numbers holds in this framework.

\begin{theorem}
\label{t:LLN_iid}
If the assumptions of theorem \ref{t:gs_iid} hold, then the following statements hold:

\begin{enumerate}
  \item For all $\mu \in \pP(\RR_+)$ and $h \in \lL$, we have
    \begin{equation*}
      \PP_{\mu}
                \left\{
                   \lim_{T \to \infty} 
                   \frac{1}{T}
                   \sum_{t=1}^T h(a_t)
                   = \int h \diff \psi^*
                \right\}
      = 1.
    \end{equation*}
  
  \item For all $\mu \in \pP(\RR_+)$, we have
    \begin{equation*}
      \PP_{\mu}
                \left\{
                   \lim_{T \to \infty}
                   \frac{1}{T}
                   \sum_{t=1}^T h(a_t)
                   = \int h \diff \psi^*
                   \; \text{ for all }
                   h \in cb \RR_+
                \right\}
       = 1.
    \end{equation*}
\end{enumerate}
\end{theorem}

\subsubsection{Case II: Markovian $\{ z_t\}_{t \geq 0}$ process}
\label{ss:gs_general}

In this case, $\{ R_t\}$ and $\{ Y_t\}$ are in general non-{\sc iid} and mutually dependent processes.\footnote{Since this framework encorporates the {\sc iid} 
    $\{z_t\}$ structure as a special case, this section provides an alternative
    ergodic theory for the {\sc iid} framework as a byproduct. By comparing the 
    assumptions of theorem \ref{t:gs_iid} and those of theorem 
    \ref{t:gs_gnl_ergo_LLN} below, we see that the latter holds without 
    assumption \ref{a:concave}, so neither of the two theories 
    is more powerful than the other. } 

We assume that the stochastic processes $\{z_t\}$ and $\{Y_t\}$ admit 
density representations denoted respectively by $p \left(z' \mid z \right)$ 
and $f_L \left( Y \mid z \right)$. Specifically, there exists a nontrivial measure $\vartheta$ on $\bB(\ZZ)$ such that
\begin{equation*}
    P(z, A) = \int_A p(z'\mid z) \vartheta (\diff z'),
    \qquad \left( A \in \bB(\ZZ), z \in \ZZ \right),
\end{equation*}
and for $\diff Y := \lambda (\diff Y)$, where $\lambda$ is the Lebesgue measure, 
\begin{equation*}
    \PP \{ Y_t \in A \mid z_t = z \} 
    = \int_A f_L (Y \mid z) \diff Y,
    \qquad \left( A \in \bB(\RR_+), z \in \ZZ \right).
\end{equation*}
%
%\begin{equation*}
%	\text{and} \quad
%	\PP \{ R_t \in A \mid z_t = z \}
%	= \int_A f_C (R \mid z) \diff R.
%\end{equation*}
%% 

\begin{assumption}
	\label{a:pos_dens}
	The following conditions hold:
	\begin{enumerate}
	  \item the support of $\vartheta$ contains a compact subset $\CC$ that has 
	      nonempty interior,\footnote{The \textit{support} of the measure 
	      	  $\vartheta$ is defined as the set of points $z \in \ZZ$ for which
	      	  every open neighborhood of $z$ has positive $\vartheta$ measure.}
	  \item $p \left( z' \mid z \right)$ is strictly positive on 
	      $\CC \times \ZZ$ and continuous in $z$, and
	  \item there exists $\delta_Y > 0$ such that $f_L \left( Y \mid z \right)$  
	      is strictly positive on $(0, \delta_Y) \times \CC$.
%	  \item there exists $\delta_R > 0$ such that $f_C \left( R \mid z \right)$ 
%	      is strictly positive on $(0, \delta_R) \times \ZZ$.
	\end{enumerate}
\end{assumption}

Assumption \ref{a:pos_dens} is easy to verify in applications. The following examples are some simple illustrations, while more complicated applications are treated in section \ref{s:app}.

\begin{example}
	\label{ex:suff_pos_dens_ctb}
	If $\ZZ$ is a countable subset of $\RR^m$, then $\{ z_t \}$ is a countable 
	state Markov chain, in which case $\vartheta$ is the counting measure and 
	$p(z'\mid z)$ reduces to a transition matrix $\Pi$. In particular, each 
	single point in $\ZZ$ is a compact subset in the support of $\vartheta$ that 
	has nonempty interior (itself), and $p$ is continuous in $z$ by definition. 
	Hence, conditions~(1)--(2) of assumption \ref{a:pos_dens} hold as long as at 
	least one column of $\Pi$ is strictly positive (i.e., each element of that  
	column is positive).
\end{example}

\begin{example}
	\label{ex:suff_pos_dens_Leb}
	Since $\ZZ$ is a Borel subset of $\RR^m$, if $\vartheta$ can be chosen as the 
	Lebesgue measure, then condition~(1) of assumption \ref{a:pos_dens} holds 
	trivially. Indeed, since $P(z, \ZZ) =1$, the support of $\vartheta$ must 
	contain a nonempty open box (i.e., sets of the form $\Pi_{i=1}^m (a_i, b_i)$ 
	with $a_i < b_i$, $i = 1, \cdots, m$), inside which a compact subset with 
	nonempty interior can be found.
\end{example}

For all measurable map $f \colon \SS \to [1, \infty)$ and $\mu \in \pP(\SS)$, 
we define
\begin{equation*}
	\left\| \mu \right\|_f 
	:= \sup_{g: |g| \leq f} 
	\left| \int g \diff \mu \right|.
\end{equation*}
We say that the stochastic kernel $Q$ corresponding to $\{ (a_t,z_t)\}_{t \geq 0}$ is \textit{$f$-ergodic} if 
\begin{enumerate}
	\item[(a)] there exists a unique stationary distribution 
	    $\psi^* \in \pP(\SS)$ such that  $\psi^* Q = \psi^*$,
	\item[(b)] $f \geq 1$, $\int f \diff \psi^* < \infty$, and, for all 
	    $(a,z) \in \SS$,
		\begin{equation*}
		    \lim_{t \to \infty} 
		    \| Q^t \left( (a,z), \cdot \right) - \psi^* \|_f = 0.
		\end{equation*}
\end{enumerate}
We say that $Q$ is \emph{$f$-geometrically ergodic} if, in addition, there exist constants $r > 1$ and $M \in \RR_+$ such that, for all $(a,z) \in \SS$, 
\begin{equation*}
	\sum_{t \geq 0} r^t 
	\left\| Q^t ((a,z), \cdot) - \psi^* \right\|_f 
	\leq M f(a,z).
\end{equation*}
In particular, if $f \equiv 1$, then $Q$ is called \textit{ergodic}/\textit{geometrically ergodic}. 

%Moreover, we say that $Q$ is \emph{$f$-uniformly ergodic} if there exists a 
%unique stationary distribution $\psi^* \in \pP(\SS)$ such that 
%$\psi^* Q = \psi^*$, and that
%%
%\begin{equation*}
%%\vertiii{Q^n - \psi^*}_f := 
%\sup_{(a,z) \in S} \frac{\| Q^t \left( 
%	\left(a,z \right), \cdot 
%	\right) 
%	- \psi^* 
%	\|_f}{f(a,z)}
%\to 0
%\quad \text{as } 
%t \to \infty.
%\end{equation*}
%%

The following theorem establishes ergodicity and the Law of Large Numbers. Notably, assumption \ref{a:concave} is {\sc not} required for these results.

\begin{theorem}
	\label{t:gs_gnl_ergo_LLN}
	If assumptions \ref{a:utility}, \ref{a:ctra_coef}--\ref{a:conti_ereuprm}, \ref{a:suff_bd_in_prob}--\ref{a:z_bdd_in_prob} and \ref{a:pos_dens} hold, then 
	\begin{enumerate}
		\item $Q$ is ergodic, in particular,
		\begin{equation*}
		    \sup_{A \in \bB(\SS)} 
		    \left|Q^t \left( (a,z), A \right) - \psi^* (A) \right| \to 0
		    \quad \text{as } \, t \to \infty.
		\end{equation*}
		\item For all $\mu \in \pP(\SS)$ and map $h: \SS \to \RR$ with $\int |h| \diff \psi^* < \infty$, 
		\begin{equation*}
		    \PP_{\mu} 
		    \left\{
		        \lim_{T \to \infty} \sum_{t=1}^T h(a_t, z_t) = \int h \diff \psi^*  
		    \right\}
		    = 1.
		\end{equation*}
	\end{enumerate}
\end{theorem}

We next show that geometric ergodicity is guaranteed under some further 
assumptions. Suppose $\{R_t\}$ admits a density representation $f_C(R \mid z)$, in other words, 
\begin{equation*}
	\PP \{ R_t \in A \mid z_t = z \}
	= \int_A f_C (R \mid z) \diff R,
	\qquad \left( A \in \bB (\RR), \; z \in \ZZ \right),
\end{equation*}
where $\diff R := \lambda (\diff R)$. Recall the {\sc iid} innovations $\{ \zeta_t\}$ and $\{ \eta_t \}$ defined by \eqref{eq:RY_func} and the compact subset $\CC \subset \ZZ$ defined by assumption 
\ref{a:pos_dens}.

\begin{assumption}
	\label{a:geo_drift_Yt}
	The following conditions hold:
	\begin{enumerate}
		\item there exists $\delta_R > 0$ such that $f_C (R \mid z)$ is strictly 
		positive on  $(0, \delta_R) \times \CC$, 
				
		\item there exist $q \in [0,1)$ and $q' \in \RR_+$ such that 
		$\EE_z Y_2 \leq q \EE_z Y_{1} + q'$ for all $z \in \ZZ$,
		
		\item the innovations $\{ \zeta_t \}$ and $\{ \eta_t\}$ are mutually 
		independent.
	\end{enumerate}
\end{assumption}

\begin{example}
	\label{ex:suff_geodrift_Y}
	If either $\{Y_t\}$ is a bounded process or $\ZZ$ is a finite set, then the second condition of assumption \ref{a:geo_drift_Yt} holds trivially. In 
	particular, if $\ZZ$ is finite, then we can let $q$ be an arbitrary number in 
	$[0,1)$ and let $q':= \sup_{z \in \ZZ} \EE_z Y_2$, which is finite by 
	assumption \ref{a:Y_sum}. More general examples are discussed in the next 
	section.
\end{example}

Let the measurable map $V \colon \SS \to [1, \infty)$ be defined by
\begin{equation}
\label{eq:V_func}
  V(a,z) := a + m \, \EE_z \hat{Y} + 1,
\end{equation}
where $m$ is a sufficiently large constant defined in the proof of theorem \ref{t:gs_gnl} below.

\begin{theorem}
	\label{t:gs_gnl}
	If assumptions \ref{a:utility}, \ref{a:ctra_coef}--\ref{a:conti_ereuprm} and 
	\ref{a:suff_bd_in_prob}--\ref{a:geo_drift_Yt} hold, then $Q$ is $V$-geometrically 
	ergodic.
\end{theorem}

\section{Applications}
\label{s:app}

We now turn to several substantial applications of the theory described above. We first illustrate how our theory can be applied to modeling capital income risk in different situations. We then provide a numerical example and study the quantitative effect of stochastic volatility and mean persistence of the wealth return process on wealth inequality.

Throughout this section, we work with the CRRA utility function defined by \eqref{eq:crra_utils}. Recall that $\gamma > 0$ is the coefficient of relative risk aversion.

\subsection{Modeling Capital Income Risk}
\label{ss:app_cir_modeling}

Suppose the income process contains both persistent and transient components (see, e.g., \cite{blundell2008consumption}, \cite{browning2010modelling},
\cite{heathcote2010macroeconomic}, \cite{kaplan2010much}, 
\cite{kaplan2012inequality}, \cite{debacker2013rising}, 
and \cite{carroll2017distribution}). In particular, we consider 
\begin{equation*}
    \label{eq:app_Yt}
  \log Y_{t} = \chi_{t} + \eta_{t}, 
\end{equation*}
%\label{
where the persistent component $\{\chi_t \}_{t \geq 0}$ is a finite-state Markov 
chain with transition matrix $\Pi_\chi$, and the transient component 
$\left\{ \eta_{t} \right\}_{t \geq 1}$ is an {\sc iid} sequence with $\EE \me^{\eta_{t}} < \infty$ and $\EE \me^{-2 \gamma \eta_{t}} < \infty$.
Moreover, $\{ \chi_t\}$ and $\{ \eta_t\}$ are mutually independent.

As a natural extension of the {\sc iid} financial return process assumed by \cite{benhabib2015wealth}, we consider $\{ R_{t}\}_{t \geq 1}$ taking  
form of
\begin{equation*}
\label{eq:app_Rt}
  \log R_{t} = \mu_t + \sigma_t \zeta_{t},
\end{equation*}
where $\{ \zeta_t\}_{t \geq 1} \iidsim N(0,1)$, $\{\mu_t\}_{t \geq 0}$ and $\{\sigma_t\}_{t \geq 0}$ are respectively finite-state Markov chains with  transition matrices $\Pi_\mu$ and $\Pi_\sigma$, $\{\sigma_t \}$ is positive, and $\{\mu_t \}$, $\{\sigma_t \}$ and $\{\zeta_t \}$ are mutually independent.\footnote{Note that $\{Y_t\}$ and 
	$\{R_t \}$ are allowed to be dependent on each other since, for example, we 
	allow $\{ \chi_t \}$ and $\{\mu_t\}$ to be mutually dependent, as we do for 
	$\{\eta_t \}$ and $\{ \sigma_t \}$, etc. }
Such a setup, as it appears, allows us to capture both mean persistence and stochastic volatility.

The state spaces of $\{\chi_t \}, \{ \mu_t \}$ and $\{ \sigma_t \}$ are  respectively (sorted in increasing order)
\begin{equation*}
    \ZZ_\chi := \{ \ell_1, \cdots, \ell_K \}, \quad
    \ZZ_{\mu} := \{i_1, \cdots, i_M \} 
    \quad \text{and} \quad 
    \ZZ_{\sigma} := \{j_1, \cdots, j_N \}.
\end{equation*}
Let $\text{diag }(\cdot)$ be the diagonal matrix created by elements in the bracket, and let
\begin{equation*}
    D_{\mu} := \text{diag} 
    \left(
        \me^{i_1}, \cdots, \me^{i_M}
    \right)
    \quad \text{and} \quad
    D_\sigma := \text{diag} 
    \left(
    \me^{j_1^2 / 2}, \cdots, \me^{j_N^2 / 2}
    \right).
\end{equation*}
Furthermore, we define the column vectors
\begin{equation*}
    V_{\mu} := \left( 
            \me^{(1-\gamma) i_1}, \cdots, \me^{(1 - \gamma) i_M} 
        \right)'
    \quad \text{and} \quad
    V_{\sigma} := \left(
            \me^{(1 - \gamma)^2 j_1^2 / 2}, \cdots, 
            \me^{(1 - \gamma)^2 j_N^2 / 2}
        \right)'.
\end{equation*}
For any square matrix $A$, let $r(A)$ be its spectral radius. We assume that
\begin{equation}
\label{eq:app1_srad_bet}
    r(\Pi_{\mu} D_\mu) \cdot r(\Pi_\sigma D_\sigma) < 1 / \beta
    \quad \; \text{and} 
\end{equation}
\begin{equation}
	\label{eq:app1_patience}
	\max \left\{
	    r(\Pi_\mu D_\mu) \cdot r(\Pi_\sigma D_\sigma) , \; 1
	\right\}
	< \left( 
	    \beta \| \Pi_{\mu} V_{\mu} \| \cdot \| \Pi_{\sigma} V_{\sigma} \|
	\right)^{-1 / \gamma}.
\end{equation}
%
%and that there exists $\alpha \in (0,1)$ that satisfies 
%%
%\begin{equation}
%\label{eq:app1_patience}
%    r(D_\mu \Pi_{\mu}) \cdot r(D_\sigma \Pi_\sigma) < 1 / (1 - \alpha)
%    \quad \text{and} \quad
%    \beta \| \Pi_{\mu} V_{\mu} \| \cdot \| \Pi_{\sigma} V_{\sigma} \| \leq (1 - \alpha)^\gamma.
%\end{equation}
%%
This problem can be placed in our framework by setting 
\begin{equation*}
    z_t := \left( \chi_t, \mu_t, \sigma_t \right)
    \quad \text{and} \quad
    \ZZ := \ZZ_{\chi} \times \ZZ_{\mu} \times \ZZ_{\sigma}.
\end{equation*}
To simplify notation, we denote $z := z_0$ and $(\chi, \mu, \sigma) := (\chi_0, \mu_0, \sigma_0)$.

\subsubsection{Optimality Results}

Since $\{ \zeta_t \} \iidsim N(0,1)$, by the Fubini theorem,
\begin{equation*}
    \beta^n \EE_z R_1 \cdots R_n 
    = \beta^n \EE_z \me^{\mu_1 + \sigma_1 \zeta_1} \cdots 
    \me^{\mu_n + \sigma_n \zeta_n}
    = \beta^n (\EE_\mu \me^{\mu_1} \cdots \me^{\mu_n} )
    (\EE_\sigma \me^{\sigma_1^2 / 2} \cdots \me^{\sigma_n^2 / 2}).
\end{equation*}
For all bounded functions $f$ on $\ZZ_\mu$ and $h$ on $\ZZ_\sigma$, we define
\begin{equation*}
    K_1 f(\mu) := \EE_\mu \me^{\mu_1} f(\mu_1)
    \quad \text{and} \quad
    K_2 h(\sigma) := \EE_\sigma \me^{\sigma_1^2 / 2} h(\sigma_1).
\end{equation*}
Similar to example \ref{ex:spec_rad_ctra}, $\beta^n \sup_{z} \EE_z R_1 \cdots R_n < 1$ for some $n \in \NN$ if and only if $\beta r(K_1) r(K_2) < 1$.\footnote{As 
	in example \ref{ex:spec_rad_ctra}, we have $\| K_1^n \| = \sup_\mu \EE_\mu \me^{\mu_1} \cdots \me^{\mu_n}$ and $\| K_2^n \| = \sup_\sigma \EE_\sigma \me^{\sigma_1^2 / 2} \cdots \me^{\sigma_n^2 / 2}$. Then $\beta r(K_1) r(K_2) < 1$ iff $\beta \|K_1^n \|^{1/n} \|K_2^n \|^{1/n} < 1$ for some $n \in \NN$ iff $\beta^n \|K_1^n \| \|K_2^n \| < 1$ for some $n \in \NN$ iff $\sup_z \beta^n \EE_z R_1 \cdots R_n < 1$ for some $n \in \NN$.} 
The latter obviously holds since \eqref{eq:app1_srad_bet} holds, and, similar to example \ref{ex:spec_rad_matrix}, $r(K_1) = r(\Pi_\mu D_\mu)$ and $r(K_2) = r(\Pi_\sigma D_\sigma)$. Assumption \ref{a:ctra_coef} is verified.

Using the fact that $\ZZ$ is a finite space, we have
\begin{equation}
\label{eq:app1_bd_chi}
    \sup_{t \geq 0} \sup_z \EE_z \, \me^{\chi_t} 
    = \sup_{t \geq 0} \sup_\chi \EE_\chi \, \me^{\chi_t} 
    \leq \sup_{t \geq 0} \sup_\chi \EE_\chi \me^{\ell_K} 
    = \me^{\ell_K} < \infty.
\end{equation}
Since in addition $\{\eta_t \}$ is {\sc iid} with $\EE \me^{\eta_t} < \infty$, we have
\begin{equation*}
\sup_{t \geq 0} \sup_z \EE_z Y_t
  = \sup_{t \geq 0} \sup_z \EE_z \me^{\chi_t + \eta_t}   
  = \left(\sup_{t \geq 0} \sup_z \EE_z \me^{\chi_t} \right) \EE \me^{\eta_1}
  < \infty.
\end{equation*}
Hence, assumption \ref{a:Y_sum} holds. As a byproduct, we have also verified assumptions \ref{a:bd_in_prob_Yt} and \ref{a:geo_drift_Yt}-(2) (recall example 
\ref{ex:suff_geodrift_Y}). Similarly, since $\sup_z \EE_z \me^{-2 \gamma \chi_1} \leq \me^{-2 \gamma \ell_1} < \infty$ and $\EE \me^{-2 \gamma \eta_t} < \infty$, we have
\begin{equation}
\label{eq:bd_uy2}
  \sup_z \EE_z \left[ u' \left( Y_1 \right) \right]^2
  = \sup_z \EE_z \me^{-2 \gamma (\chi_1 + \eta_1)}
  = \left(\sup_z \EE_z \me^{-2 \gamma \chi_1} \right) 
  \EE \me^{-2 \gamma \eta_1}< \infty.
\end{equation}
Moreover, for all $z \in \ZZ$, based on the Fubini theorem, 
\begin{align*}
    \EE_z \hat{R}^2 
    = \EE_z \me^{ 2 \mu_1 + 2 \sigma_1 \zeta_1}
    = \EE_\mu \me^{2 \mu_1} \EE_\sigma \me^{2 \sigma_1 \zeta_1}
    = \EE_\mu \me^{2 \mu_1} \EE_\sigma \me^{2 \sigma_1^2}
    \leq \me^{2 i_M + 2 j_N^2 } < \infty.
\end{align*}
Hence, assumption \ref{a:bd_sup_ereuprm} holds (see example \ref{ex:CS_suff}). Since $\ZZ$ is a finite space, this in turn implies that $z \mapsto \EE_z \hat{R} u' (\hat{Y})$ must be continuous, so assumption \ref{a:conti_ereuprm} holds.

In summary, we have verified all the assumptions of section \ref{s:opt_results}. All
the related optimality results have been established.

\subsubsection{Existence of Stationary Distributions}

Similar to examples \ref{ex:homog}--\ref{ex:bdd_in_prob_matrix}, assumption \ref{a:suff_bd_in_prob}-(1) holds if 
$(1 - \alpha)^{-\gamma} \beta \EE_z \hat{R}^{1 - \gamma} \leq 1$ for all 
$z$. Since 
\begin{align*}
    \EE_z \hat{R}^{1 - \gamma}
    &= \EE_{\sigma} 
      \me^{ (1-\gamma) (\mu_1 + \sigma_1 \zeta_1) }
    = \EE_{\mu} \me^{ (1-\gamma) \mu_1} 
    \EE_{\sigma} \me^{ (1 - \gamma) \sigma_1 \zeta_1}    \\
    &= \EE_{\mu} \me^{ (1-\gamma) \mu_1} 
    \EE_{\sigma} \me^{ (1 - \gamma)^2 \sigma_1^2 / 2}
    \leq \| \Pi_{\mu} V_{\mu} \| \cdot \| \Pi_{\sigma} V_\sigma \|,
\end{align*}
it suffices to show that $\beta \| \Pi_{\mu} V_{\mu} \| \cdot \| \Pi_{\sigma} V_\sigma \| \leq (1 - \alpha)^\gamma$. 
Moreover, similar to verifying assumption \ref{a:ctra_coef}, assumption \ref{a:suff_bd_in_prob}-(2) holds as long as $(1 - \alpha) r(\Pi_\mu D_\mu) r(\Pi_\sigma D_\sigma) < 1$. In summary, assumption \ref{a:suff_bd_in_prob} holds whenever there exists $\alpha \in (0,1)$ that satisfies
\begin{equation*}
    r(\Pi_\mu D_\mu) \cdot r(\Pi_\sigma D_\sigma)
    < 1 / (1 - \alpha) 
    \leq \left( 
         \beta  \| \Pi_{\mu} V_{\mu} \| \cdot \| \Pi_{\sigma} V_\sigma \|  
    \right)^{-1 / \gamma}.
\end{equation*}
This is guaranteed by \eqref{eq:app1_patience}. Moreover, assumption \ref{a:bd_in_prob_Yt} has been verified in the previous section, assumption \ref{a:z_bdd_in_prob} is trivial since $\ZZ$ is finite, and assumption 
\ref{a:concave} has been verified in example \ref{eg:concave}.

In summary, all the assumptions up to section \ref{ss:exist_stat} have been verified. As a result, all the conclusions of propositions \ref{pr:opt_pol_bd_frac}--\ref{pr:optpol_concave} and theorem \ref{t:sta_exist} hold.

\subsubsection{Global Stability}

Regarding ergodicity and the Law of Large Numbers (theorem 
\ref{t:gs_gnl_ergo_LLN}), it remains to verify assumption \ref{a:pos_dens}. 
This is true if we assume further
\begin{itemize}
	\item there are strictly positive columns in each of the matrices $\Pi_\chi$, 
	    $\Pi_{\mu}$ and $\Pi_{\sigma}$ (recall example 
	    \ref{ex:suff_pos_dens_ctb}), and
	\item $\{\eta_t\}$ has a density that is strictly positive on 
	    $(-\infty, \delta)$ for some $\delta \in \RR$.
\end{itemize}
Regarding geometric ergodicity (theorem \ref{t:gs_gnl}), it remains to verify assumption \ref{a:geo_drift_Yt}. Condition (1) is trivial since $\{ \zeta_t\} \iidsim N(0,1)$. Condition (2) has been verified in previous sections. Hence, the model is $V$-geometrically ergodic as long as the innovations 
$\{ \eta_t\}$ and $\{\zeta_t\}$ are mutually independent.

\subsection{Modeling Generic Stochastic Returns}

Indeed, our theory works for more general setups. To illustrate, consider the following labor income process\footnote{Similar extensions can be made to the $\{R_t\}$ process.}
\begin{align}
	\label{eq:app2_Yt}
	Y_{t} = \chi_{t} \,  \varphi_{t} + \nu_{t} 
	\quad  \text{and}  \quad
	\ln \chi_{t+1} = \rho \ln \chi_{t} + \epsilon_{t+1},
\end{align}
where $\chi_0 \in (0, \infty)$ and $\rho \in (0,1)$ are given, 
$\left\{ \epsilon_t \right\}_{t \geq 1} \iidsim N(0, \delta^2)$, 
$\{ \nu_t \}_{t \geq 1}$ and $\{\varphi_t \}_{t \geq 1}$ are positive {\sc iid} 
sequences with finite second moments, and $\EE \nu_{t}^{-2 \gamma} < \infty$. 
Moreover, $\{\chi_t \}$, $\{\varphi_t \}$ and $\{\nu_t \}$ are mutually 
independent. Similar setups appear in a lot of applied literature. 
See, for example,
\cite{heathcote2010macroeconomic}, \cite{kaplan2010much}, 
\cite{huggett2011sources}, \cite{kaplan2012inequality} and 
\cite{debacker2013rising}.

This setup can be placed in our framework by setting 
$\eta_t := (\varphi_{t}, \nu_{t})$. Next, we aim to verify all the assumptions 
related to $\{ Y_t \}$.

Based on \eqref{eq:app2_Yt}, for all $t \geq 0$, the distribution of 
$\chi_t$ given $\chi_0$ follows 
\begin{equation*}
	\left( \chi_t \mid \chi_0 \right) 
	\sim 
	LN \left( 
	    \rho^t \ln \chi_0, \;
	    \delta^2 \sum_{k=0}^{t-1} \varphi^{2k}
	\right).
\end{equation*}
We denote $\chi := \chi_0$ for simplicity. Then for all $t \geq 0$ and $s \in \RR$, we have\footnote{Recall that for $X \sim LN(\mu, \sigma^2)$ and $s \in \RR$, we have 
	$\EE (X^s) = \exp \left( s \mu + s^2 \sigma^2 / 2 \right)$.}
\begin{align*}
	\label{eq:mgf}
	\EE_{\chi} \chi_t^s = \exp 
	\left[ 
	    s \rho^t \ln \chi + 
	    \frac{s^2 \delta^2 (1 - \rho^{2t})}{2 (1 - \rho^2)}
	\right].
\end{align*}
In particular, since $\rho \in (0,1)$, this implies that $\sup_{t \geq 0} \EE_\chi \chi_t^s < \infty$ for all $s \in \RR$ and $\chi \in (0, \infty)$. Hence,
\begin{equation*}
	\sup_{t \geq 0} \EE_\chi Y_t
	= \sup_{t \geq 0} \EE_\chi \chi_t \varphi_{t} + \EE \nu_{t}   
	\leq \left(\sup_{t \geq 0} \EE_\chi \chi_t \right) \EE \varphi_{t} +
	    \EE \nu_{t} 
	< \infty
\end{equation*}
for all $\chi \in (0, \infty)$, and assumptions \ref{a:Y_sum} and \ref{a:bd_in_prob_Yt} hold. Moreover, since $Y_t \geq \nu_{t}$, 
\begin{equation*}
	\sup_\chi \EE_\chi \left[ u' \left( Y_t \right) \right]^2
	\leq \EE \left[ u' \left( \nu_{t} \right) \right]^2 
	= \EE \nu_{t}^{-2 \gamma}< \infty,
\end{equation*}
and the second part of assumption \ref{a:bd_sup_ereuprm} holds.
Regarding assumption \ref{a:geo_drift_Yt}-(2), since $\rho \in (0,1)$, we can choose $\bar{\chi} > 0$ such that
\begin{equation*}
    q := \me^{\delta^2 \rho^2 / 2} \bar{\chi}^{\rho (\rho - 1)} < 1.
\end{equation*}
Then for $\chi \leq \bar{\chi}$, we have $\EE_\chi \chi_2 \leq \me^{ \rho^2 \ln \bar{\chi} + \delta^2(1+ \rho^2)/2 } =: d$, and for $\chi > \bar{\chi}$, we have
\begin{align*}
	\EE_\chi \chi_2
	&= \me^{ \delta^2(1+ \rho^2)/2} \chi^{\rho^2}   
	= \frac{ \me^{ \delta^2(1+ \rho^2)/2} \chi^{\rho^2}  
	    }{\me^{\delta^2 / 2 } \chi^\rho} 
	    \cdot \me^{\delta^2 / 2 } \chi^\rho   \\
	&= \me^{ \delta^2 \rho^2 / 2} \chi^{\rho (\rho - 1)} 
	    \cdot \EE_\chi \chi_1    
	\leq \me^{ \delta^2 \rho^2 / 2} \bar{\chi}^{\rho (\rho - 1)} 
	\cdot \EE_\chi \chi_1  
	= q \, \EE_z \chi_1.
\end{align*}
Hence, $\EE_\chi \chi_2 \leq q \, \EE_\chi \chi_1 + d$ for all $\chi$. Since in addition $\EE \varphi_t < \infty$, $\EE \nu_t < \infty$ and 
\begin{equation*}
    \EE_\chi Y_2 = \EE_\chi \chi_2 \, \EE \varphi_2 + \EE \nu_2,
\end{equation*}
assumption \ref{a:geo_drift_Yt}-(2) follows immediately.

Finally, assumption \ref{a:pos_dens}-(3) holds as long as the distributions of 
$\{ \varphi_ t\}$ and $\{ \nu_t\}$ have densities that are strictly positive on 
$(0, \bar{\delta})$ for some $\bar{\delta} > 0$.

\subsection{Numerical Example}
\label{ss:app_numerical}

What are the ``wealth inequality effects" of mean persistence and stochastic 
volatility in the rate of return to wealth? This is an important question that is rarely 
explored by the existing literature. In what follows we attempt to provide an
answer via simulation. In doing this, we will also explore the generality of our 
theory by testing the stability properties of the economy for a broad range of 
parameters. Our study is based on the model of section \ref{ss:app_cir_modeling}.

Regarding the finite-state Markov chains $\{ \chi_t \}$, $\{ \mu_t\}$ and 
$\{ \eta_t\}$, we use the method of \cite{tauchen1991quadrature} 
and discretize the following AR(1) processes
\begin{align*}
    \chi_t &= \rho_\chi \chi_{t-1} + \epsilon_t^{\chi}, 
        \qquad \{ \epsilon_t^\chi \} \iidsim N(0, \delta_\chi^2),    \\ 
    \mu_t &= (1 - \rho_\mu) \bar{\mu} + \rho_\mu \mu_{t-1} + \epsilon_t^\mu ,
        \qquad \{ \epsilon_t^\mu \} \iidsim N(0, \delta_\mu^2),    \\
    \log \sigma_t &= (1- \rho_\sigma) \bar{\sigma} 
        + \rho_\sigma \log \sigma_{t-1}
        + \epsilon_t^\sigma, 
         \qquad \{ \epsilon_t^\sigma \} \iidsim N(0,\delta_\sigma^2 ),
\end{align*}
into $N_\chi	$, $N_\mu$ and $N_\sigma$ states, respectively. 

Regarding the parameters of the $\{Y_t\}$ process, we set $\{ \eta_t\}$ to
be a normal distribution with mean $0$ and variance $\delta_\eta^2 = 0.075$.
In addition, we set $\rho_\chi = 0.9770$ and $\delta_\chi^2 = 0.02$. These 
values are chosen broadly in line with the existing literature. See, for example,
\cite{heathcote2010macroeconomic}, \cite{kaplan2010much}, and 
\cite{debacker2013rising}.

Our calibration of the $\{R_t\}$ process is based on  
\cite{fagereng2016heterogeneity}, in which the authors report the average and 
standard deviation of the financial return process of Norway from 1993--2013.\footnote{This is the only data source we can find that has a full record of financial 
    returns. Although our calibration is based on this dataset, we have conducted
    sensitivity analysis for different groups of parameters. The results show that
    their qualitative effects are broadly the same, although their quantitative effects
    vary, as one would expect.} 
We transform the two series to match our model and run first-order 
autoregressions, which yield
$\bar{\mu} = 0.0281$, $\rho_\mu = 0.5722$, $\delta_\mu = 0.0067$, 
$\bar{\sigma}=-3.2556$, $\rho_\sigma=0.2895$ and $\delta_\sigma=0.1896$. 
Based on this parameterization, the stationary mean and standard deviation of the 
$\{R_t\}$ process are approximately $1.03$ and $4\%$, respectively.

However, to distinguish the different effect of stochastic volatility and mean persistence, as well as to mitigate the computational burden caused by high state dimensionality, 
we consider two subsidiary model economies. The first model reduces 
$\{ \mu_t\}$ to its stationary mean $\bar{\mu}$, while the second model reduces 
$\{ \sigma_t \}$ to its stationary mean 
$\hat{\sigma} := \exp (\bar{\sigma} + \delta_\sigma^2/ 2(1 - \rho_\sigma^2))$. 
In summary, $\{R_t\}$ satisfies
\begin{align*}
    &\log R_t = \bar{\mu} + \sigma_t \zeta_t  
    \qquad (\text{Model \rom{1}})    \\
    &\log R_t = \mu_t + \hat{\sigma} \zeta_t
    \qquad (\text{Model \rom{2}})
\end{align*}
To test the stability properties of the economy, we set $\beta=0.95$, $N_\chi = 5$ and consider respectively $\gamma=1$ and $\gamma=2$. Furthermore, in model \rom{1}, we set $N_\sigma=5$ and consider a broad neighborhood of the calibrated $(\rho_\sigma, \delta_\sigma)$ pairs, and in model \rom{2}, we set $N_\mu=5$ and consider a large neighborhood around the calibrated $(\rho_\mu, \delta_\mu)$ values. Each scenario, we hold the rest of the parameters as in the benchmark. The results are shown in figure \ref{fig:m1} and figure \ref{fig:m2}. 

Since the dot points (calibrated parameter values) lie in the stable range in all cases, both the two calibrated models are globally stable, and stationary wealth distributions can be computed by the established ergodic theorems (theorem \ref{t:gs_gnl_ergo_LLN} and theorem \ref{t:gs_gnl}). Moreover, the broad stability range indicates that our theory can handle a wide range of parameter setups, including highly persistent and volatile $\{R_t\}$ processes.

\begin{figure}
\centering
\begin{subfigure}[a]{0.75\textwidth}
   \includegraphics[width=1\linewidth]{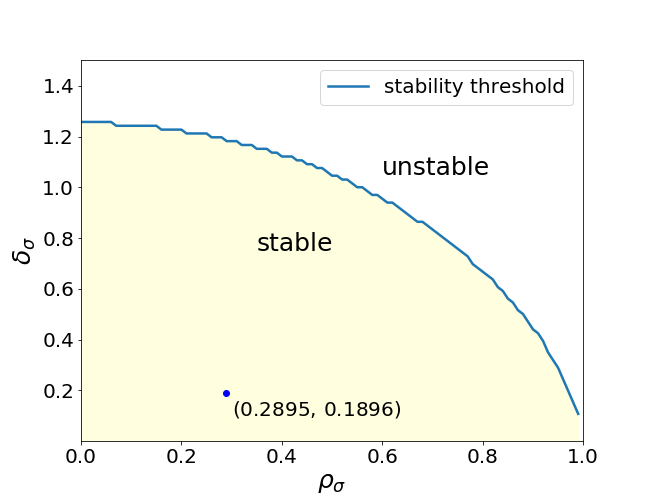}
   \caption*{(a) Model \rom{1} : $\beta = 0.95$, \, $\gamma = 1$, \, $\bar{\mu} = 0.0281$}
   \label{fig:m1_gam1} 
\end{subfigure}
\vspace{0.cm}
\begin{subfigure}[b]{0.75\textwidth}
   \includegraphics[width=1\linewidth]{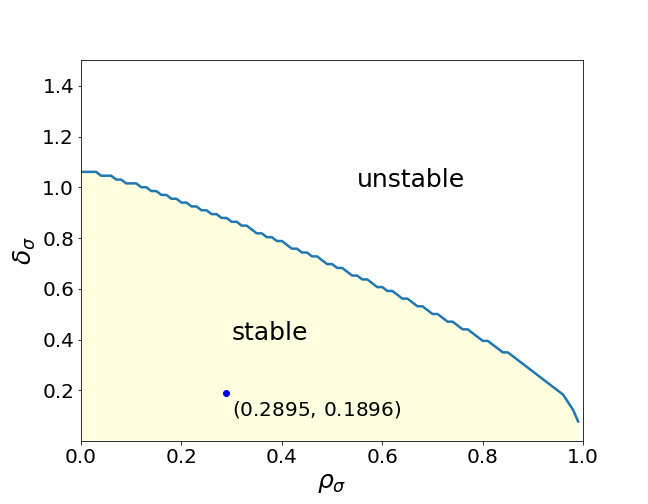}
   \caption*{(b) Model \rom{1} : $\beta = 0.95$, \, $\gamma = 2$, \, $\bar{\mu}=0.0281$}
   \label{fig:m1_gam2}
\end{subfigure}
\caption[]{Stability Range and Threshold of Model \rom{1}}
\label{fig:m1}
\end{figure}

\begin{figure}
\centering
\begin{subfigure}[a]{0.8\textwidth}
   \includegraphics[width=1\linewidth]{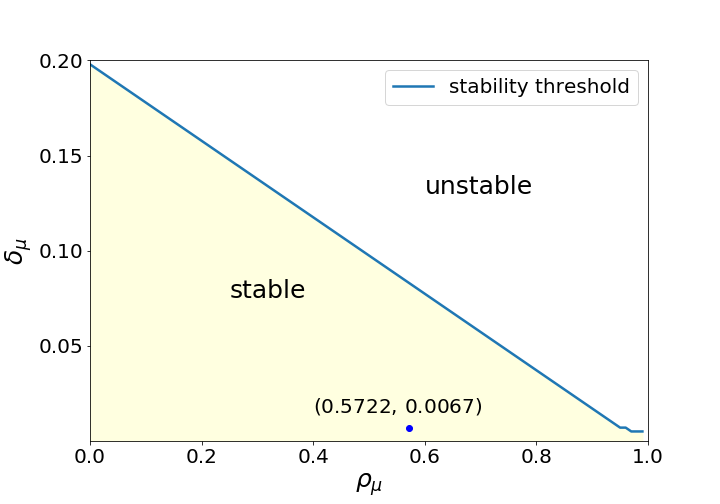}
   \caption*{(a) Model \rom{2} : $\beta = 0.95$, \, $\gamma = 1$, \, $\hat{\sigma} = 0.0393$}
   \label{fig:m1_gam1} 
\end{subfigure}
\vspace{0.cm}
\begin{subfigure}[b]{0.8\textwidth}
   \includegraphics[width=1\linewidth]{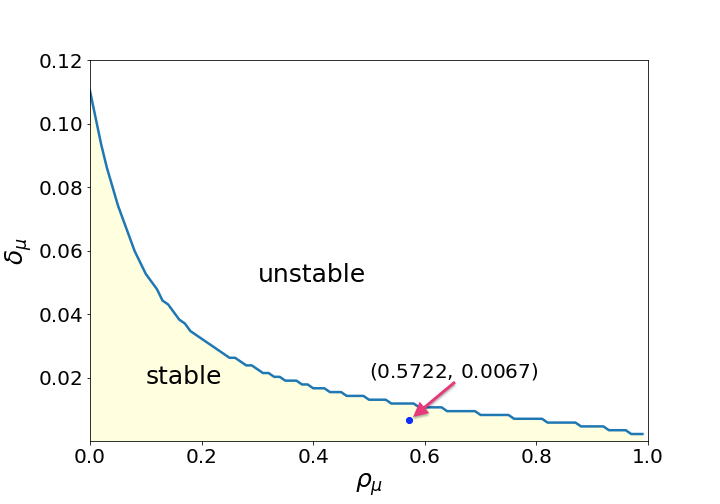}
   \caption*{(b) Model \rom{2} : $\beta = 0.95$, \, $\gamma = 2$, \, $\hat{\sigma}=0.0393$}
   \label{fig:m1_gam2}
\end{subfigure}
\caption[]{Stability Range and Threshold of Model \rom{2}}
\label{fig:m2}
\end{figure}

Our next goal is to explore the quantitative impact of capital income risk on wealth 
inequality. As a first step, we compute the optimal policy. This can be realized by  
iterating the Coleman opeartor and evaluating the distance between loops via the 
designed metric $\rho$. The algorithm is guaranteed to converge based on 
theorem \ref{t:ctra_T}. Specifically, we assign 100 grid points to wealth equally 
spaced in $[10^{-4}, 50]$. Expectations with respect to the {\sc iid} innovations are 
evaluated via Monte Carlo with $1000$ draws. Moreover, in all cases, we use 
piecewise linear interpolation to approximate policies. Policy function evaluation 
outside of the grid range is via linear extrapolation, as is justified by proposition 
\ref{pr:optpol_concave}. 

Once the optimal policy is obtained, we then simulate a single time series of $5 \times 10^7$ agents in each case and compute the stationary distribution based on our ergodic theorems \ref{t:gs_gnl_ergo_LLN}--\ref{t:gs_gnl}. As a final step, we compare the key properties of the stationary wealth distributions in different economies. In particular, we estimate the tail exponent based on the wealth level of the top $5\%$ and top $10\%$ of the simulated agents.\footnote{Recall that a random variable $X$ is said to have a 
    \textit{heavy upper tail} if there exist constants $A, \alpha > 0$ such that
    $\PP \{ X > x\} \geq Ax^{-\alpha}$ for large enough $x$, where $\alpha$
    is refered to as the \textit{tail exponent}. The smaller the tail exponent is,
    the fatter the distribution tail is, and thus a higher level of inequality exists.
     It is common in the literature to 
    estimate the tail exponent via linearly regressing the $\log$-ranks over the 
    $\log$-wealth levels of the top $5\%$ and top $10\%$ most wealthy agents.
    } 
Moreover, we estimate the Gini coefficient and provide a detailed analysis of the 
wealth share in each case. 

All simulations are processed in a standard Julia 
environment on a laptop with a 2.9 GHz Intel Core i7 and 32GB RAM.

\begin{table}[h]
    \caption{Tail Exponent and Gini Coefficient}
    \label{tb:te}
    \vspace*{-0.45cm}
	\noindent 
    \begin{center}		
	\begin{threeparttable}
        {\small
		\begin{tabular}{|c|c|c|c|c|c|}
		\hline 
		\multicolumn{2}{|c|}{Model Economy} & Model I & Model II & IID $\{R_{t}\}$ & Constant $\{R_{t}\}$\tabularnewline
		\hline 
		\hline 
		\multirow{2}{*}{Tail Exponent} & Top 5\% & 3.0 & 2.9 & 4.4 & 4.4\tabularnewline
		\cline{2-6} 
		 & Top 10\% & 2.6 & 2.5 & 3.7 & 3.7\tabularnewline
		\hline 
		\multicolumn{2}{|c|}{Gini Coefficient} & 0.47 & 0.45 & 0.34 & 0.33\tabularnewline
		\hline 
		\end{tabular}
    }
	\begin{tablenotes}
      \fontsize{9pt}{9pt}\selectfont
      \item Parameters: $\beta=0.95$, $\gamma=2$,
      $\bar{\mu}=0.0281$, $\bar{\sigma}=-3.2556$, 
      $\rho_\sigma=0.2895$, $\delta_\sigma=0.1896$,
      $\rho_\mu = 0.5722$ and $\delta_\mu = 0.0067$.
    \end{tablenotes}
	\end{threeparttable}
	\par\end{center}
\end{table}

\begin{table}[h]
    \caption{Wealth Share (in percentage)}
    \label{tb:wealth_share}
    \vspace*{-0.45cm}
	\noindent 
    \begin{center}		
	\begin{threeparttable}
        {\small
		\begin{tabular}{|c|c|c|c|c|c|c|c|c|c|c|}
		\hline 
		Poorest agents (\%) & 5\% & 10\% & 15\% & 20\% & 25\% & 30\% & 35\% & 40\% & 45\% & 50\%\tabularnewline
		\hline 
		\hline 
		Model I & 0.8 & 1.8 & 3.1 & 4.6 & 6.2 & 8.2 & 10.4 & 12.9 & 15.7 & 18.7\tabularnewline
		\hline 
		Model II & 1.1 & 2.4 & 3.9 & 5.7 & 7.6 & 9.7 & 12.1 & 14.7 & 17.5 & 20.6\tabularnewline
		\hline 
		IID $\{R_{t}\}$ & 1.5 & 3.4 & 5.6 & 8.0 & 10.6 & 13.4 & 16.5 & 19.8 & 23.4 & 27.3\tabularnewline
		\hline 
		Constant $\{R_{t}\}$ & 1.6 & 3.5 & 5.6 & 8.0 & 10.7 & 13.5 & 16.6 & 20.0 & 23.6 & 27.5\tabularnewline
		\hline 
		\hline 
		Poorest agents (\%) & 55\% & 60\% & 65\% & 70\% & 75\% & 80\% & 85\% & 90\% & 95\% & 100\%\tabularnewline
		\hline 
		\hline 
		Model I & 22.1 & 25.9 & 30.0 & 34.7 & 40.3 & 47.0 & 55.1 & 64.8 & 77.0 & 100\tabularnewline
		\hline 
		Model II & 24.1 & 27.8 & 31.9 & 36.6 & 42.0 & 48.5 & 56.3 & 65.7 & 77.5 & 100\tabularnewline
		\hline 
		IID $\{R_{t}\}$ & 31.4 & 35.9 & 40.7 & 46.0 & 51.8 & 58.4 & 65.7 & 74.2 & 84.3 & 100\tabularnewline
		\hline 
		Constant $\{R_{t}\}$ & 31.6 & 36.1 & 41.0 & 46.3 & 52.0 & 58.5 & 65.9 & 74.3 & 84.4 & 100\tabularnewline
		\hline 
		\end{tabular}
    }
	\begin{tablenotes}
      \fontsize{9pt}{9pt}\selectfont
      \item Parameters: same as table \ref{tb:te}. In the first and sixth rows,
      $N\%$ denotes the $N\%$ of agents with lowest levels of wealth.
    \end{tablenotes}
	\end{threeparttable}
	\par\end{center}
\end{table}

\begin{figure}
\centering
   \includegraphics[width=1\linewidth]{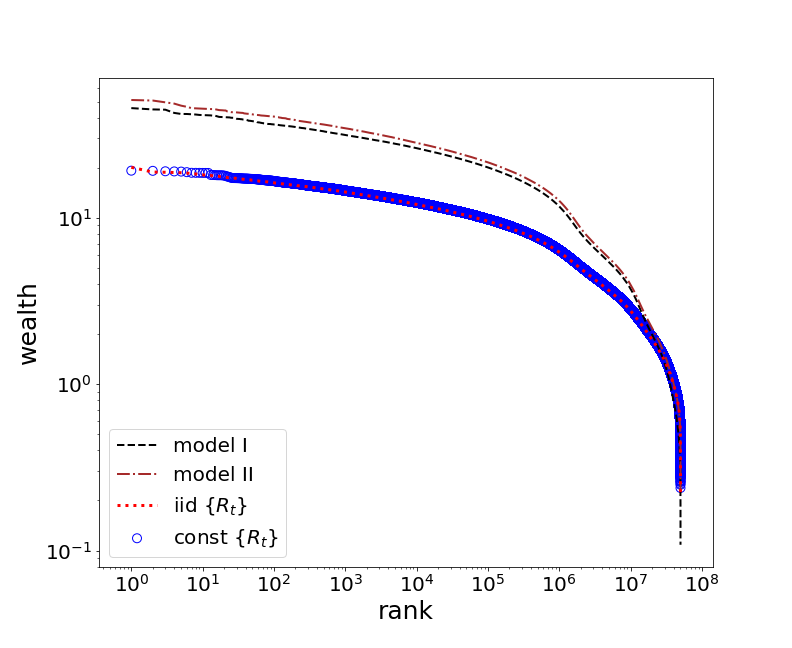}
   \caption{The Zipf Plot}
   \label{fig:zipf}
\end{figure}

\begin{figure}
\centering
   \includegraphics[width=1\linewidth]{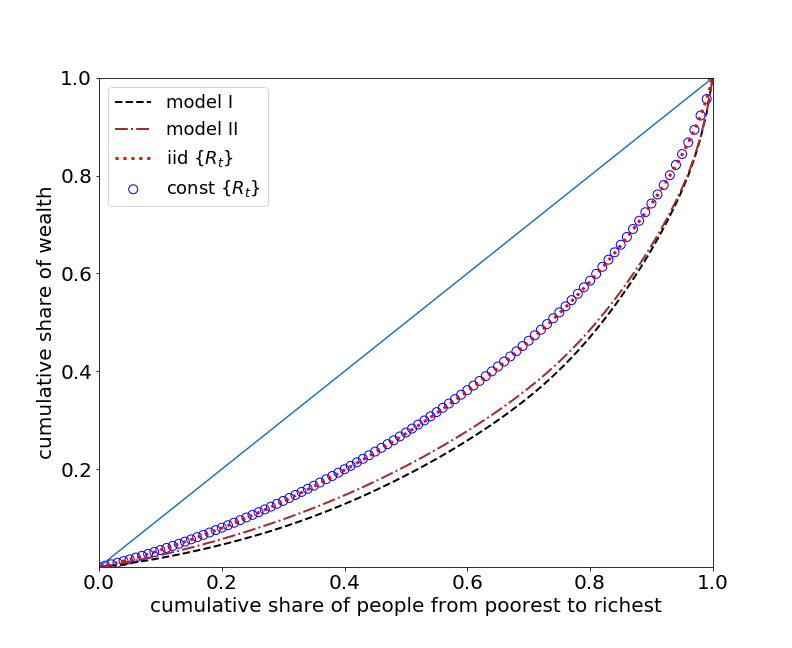}
   \caption{The Lorenz Curve}
   \label{fig:lorenz}
\end{figure}

We compare our models with two other models,
%those of \cite{aiyagari1994uninsured} and \cite{benhabib2015wealth}, 
in which $\{ R_t \}$ is respectively an {\sc iid} process and a constant.\footnote{In 
    the former case, we set $N_\sigma=1$ in model \rom{1} (so that 
    $\sigma_t$ reduces to its stationary mean) or $N_\mu=1$ in model \rom{2} 
    (so that $\mu_t$ reduces to its stationary mean). In the latter case, we reduce 
    $\{R_t\}$ to its
    stationary mean.} 
%\footnote{To obtain an analogous model to that of
%    \cite{benhabib2015wealth}, we set $N_\sigma=1$ in model \rom{1} (so that 
%    $\sigma_t$ reduces to its stationary mean) or $N_\mu=1$ in model \rom{2} 
%    (so that $\mu_t$ reduces to its stationary mean). To reach an Aiyagari-style 
%    economy as that presented in
%    \cite{aiyagari1994uninsured}, we reduce $\{R_t\}$ to its
%    stationary mean.} 
The difference between the results of model \rom{1} and model \rom{2} and the results of the other two models reflects the role of 
stochastic volatility and mean persistence of the wealth return process.
Parameter setups and results are reported in 
tables \ref{tb:te}--\ref{tb:wealth_share}.\footnote{Since the standard 
     Bewley-Ayagari-Hugget model does not generate fat-tailed wealth distribution
     (see, e.g., \cite{stachurski2018impossibility}),
     calculating the tail exponent of the stationary wealth distribution when 
     $\{R_t\}$ is a constant is relatively less standard. 
     However, doing this allows us to reveal the effect of capital income 
     risk on the tail thickness of the stationary wealth distribution.} 

As can be seen in table \ref{tb:te}, the tail exponents of model \rom{1} and 
model \rom{2} are smaller than the tail exponents when $\{R_t\}$ is {\sc iid} or 
constant. In other words, both stochastic volatility and mean persistence in 
wealth returns lead to a higher degree of wealth inequality. Moreover, mean 
persistence results in slightly lower tail exponents than stochastic volatility does.

Similarly, the Gini coefficients generated by model \rom{1} and model \rom{2}
are much higher than those generated by the other two models, illustrating from another perspective that stochastic 
volatility and mean persistence of wealth returns cause more inequality in wealth. 
However, different from the previous case, compared with mean persistence, which 
generates a Gini index 0.45, stochastic volatility has a higher impact on wealth 
inequality, creating a Gini index 0.47.

Moreover, at least in the current models, {\sc iid} wealth returns do not have  
obvious effect on wealth inequality, both in terms of their impact on the tail 
exponent and in terms of their impact on the Gini coefficient.

The above descriptions are further illustrated in table \ref{tb:wealth_share} and 
figures \ref{fig:zipf}--\ref{fig:lorenz}. In particular, in table \ref{tb:wealth_share} 
we calculate the wealth share of a given fraction of poorest agents. Notably, the top 
$10\%$ richest agents hold respectively $35.2\%$, $34.3\%$, $25.8\%$ and 
$25.7\%$ of the total wealth, while the poorest $10\%$ agents hold respectively 
$1.8\%$, $2.4\%$, $3.4\%$, $3.5\%$ of the total wealth  in the four model economies. In 
figure \ref{fig:zipf} we create the Zipf plot (i.e., plotting $\log$ wealth v.s. 
$\log$ rank). 
It is clearly indicated that model \rom{1} and model \rom{2} generate stationary 
wealth distributions with fatter upper tails than the other models do, and that 
the stationary wealth distribution of model \rom{2} has the fattest upper tail. In figure \ref{fig:lorenz} we plot the Lorenz curve, 
which can be viewed as a generalized graphical representation of table 
\ref{tb:wealth_share}.

Finally, sensitivity analysis with respect to model parameters and a more detailed quantitative analysis can be found in the online appendix of this paper.

\newpage

\section{Appendix A: Proof of Section \ref{s:opt_results} Results}
\label{Appendix_A}

In proofs we let $\{\fF_t\}_{t \geq 0}$ be the natural filtration, where $\fF_t := \sigma(s_0, \cdots, s_t)$ with $s_t := (a_t, z_t)$ for all $t$. We start by proving the results of section \ref{s:opt_results}. 

\begin{proof}[Proof of example \ref{ex:spec_rad_ctra}]
	Note that for fixed $n \in \NN$,
	\begin{equation*}
		\| K^n \| = \sup_{ \|f\| \leq 1 } \| K^n f \| 
		= \sup_{ \|f\| \leq 1 } 
		\sup_{z \in \ZZ} 
		\left| 
		    \EE_z R_1 \cdots R_n f(z_n) 
		\right|
		= \sup_{z \in \ZZ} \EE_z R_1 \cdots R_n.
	\end{equation*}
	Suppose assumption \ref{a:ctra_coef} holds. Note that every $t \in \NN$ can be written as $t = kn + \ell$ where $k \in \NN \cup \{0\}$ and $\ell \in \{0, \cdots, n-1 \}$. Since $\| K^t \| = \| K^{kn + \ell} \| \leq \|K^n\|^k \|K^\ell \|$, 
	\begin{equation*}
		\| K^t \|^{1/t} 
		= \| K^{kn + \ell} \|^{1/t} 
		\leq \|K^n\|^{k/t} \|K^\ell \|^{1/t}
		= \|K^n \|^{\frac{1}{n + \ell / k}} \| K^\ell \|^{1/t}.
	\end{equation*}
	Since $\beta \|K^n \| < 1$ by assumption \ref{a:ctra_coef} and 
	$\|K^\ell \| \leq \| K\|^\ell <\infty$, letting $t \to \infty$ (and thus 
	$k \to \infty$) yields
	\begin{equation*}
		\beta r(K) = \beta \lim_{t \to \infty} \|K^t \|^{1/t} 
		\leq \beta \lim_{k \to \infty} \|K^n \|^{\frac{1}{n} \frac{n}{n + \ell / k}} \| K^\ell \|^{\frac{1}{kn + \ell}}
		= \beta \|K^n \|^{1/n} < 1.
	\end{equation*}
	On the other hand, suppose $\beta r(K) < 1$. Then by the definition of $r$ 
	there exists $n \in \NN$ such that $\beta \| K^n \|^{1/n} < 1$. Thus 
	$\beta^n \|K^n \| <1$ and assumption \ref{a:ctra_coef} is verified.
\end{proof}

For the rest of this section, we let $n$ and $\theta$ be defined as in assumption \ref{a:ctra_coef}.

\begin{proof}[Proof of lemma \ref{lm:max_path}]
Iterating backward on the maximal path \eqref{eq:max_path}, we can show 
that
\begin{equation*}
    \tilde{a}_t = 
        \left( \prod_{i=1}^t R_i \right) a + 
        \sum_{j=1}^t \left( Y_j \, \prod_{i=j+1}^t R_i \right). 
\end{equation*}
Taking discounted expectation yields
\begin{align*}
    \beta^t \EE_{a,z} \tilde{a}_t 
    &=
      \left[ \EE_{z} \left( \beta^t \prod_{i=1}^t R_i \right) \right] a + 
      \sum_{j=1}^t       
          \EE_{z} 
          \left[            
              \left( \beta^{t-j} \prod_{i=j+1}^t R_i \right)
              \left( \beta^j Y_j \right)   
          \right].   
\end{align*}
Let $M(a,z) : = \sum_{t \geq 0} \beta^t \EE_{a,z} \tilde{a}_t$. Then the monotone convergence theorem and the Markov property imply that
\begin{align*}
  M(a,z) &=
  \sum_{t=0}^{\infty} \EE_z \left( \beta^t \prod_{i=1}^t R_i \right) a +
  \sum_{t=0}^{\infty} \sum_{j=1}^t       
  \EE_{z} 
  \left[            
  \left( \beta^{t-j} \prod_{i=j+1}^t R_i \right)
  \left( \beta^j Y_j \right)   
  \right]    \\
  &= \EE_z \left( \sum_{t=0}^{\infty} \beta^t \prod_{i=1}^t R_i \right) a +   
  \sum_{j=1}^{\infty} \EE_z \EE_z
  \left[ (\beta^j Y_j) 
      \left( 
          \sum_{i = 0}^{\infty} \beta^i \prod_{k=1}^{i} R_{j+k}
      \right)
      \Big{ | } \fF_{j}
  \right]    \\
  &= \EE_z \left( \sum_{t=0}^{\infty} \beta^t \prod_{i=1}^t R_i \right) a +   
  \sum_{j=1}^{\infty} \EE_z 
  \left[ (\beta^j Y_j) \,
  \EE_{z_j}
  \left( 
  \sum_{i = 0}^{\infty} \beta^i \prod_{k=1}^{i} R_{k}
  \right)
  \right].
\end{align*}
By the Markov property and assumption \ref{a:ctra_coef}, for all $k \in \NN$ and $z \in \ZZ$, we have
\begin{align*}
    \EE_z \beta^{kn} R_1 \cdots R_{kn} 
    &= \EE_z \EE_z [\beta^{(k-1)n} R_1 \cdots R_{(k-1)n} 
    \beta^n R_{(k-1)n+1} \cdots R_{kn} \mid \fF_{(k-1)n}]    \\
    &= \EE_z \beta^{(k-1)n} R_1 \cdots R_{(k-1)n} 
    \EE_{z_{(k-1)n}} (\beta^n R_{1} \cdots R_{n})    \\
    & \leq \theta \EE_z \beta^{(k-1)n} R_1 \cdots R_{(k-1)n} 
    \leq \cdots \leq \theta^k.
\end{align*}
Taking supremum on both sides yields
\begin{equation}
\label{eq:n-step-ctra}
    \beta^{kn} \sup_{z \in \ZZ} \EE_z R_1 \cdots R_{kn} \leq \theta^k.
\end{equation}
Moreover, assumption \ref{a:bd_sup_ereuprm} implies that 
$K_0 := \sup_{z \in \ZZ} \EE_z \hat{R} < \infty$. Hence, 
\begin{align*}
    \EE_{z} \left( \sum_{i=0}^{n-1} \beta^{i} R_{1} \cdots R_{i} \right)
    &= \sum_{i=0}^{n-1} \beta^i \EE_z R_1 \cdots R_i
    = \sum_{i=0}^{n-1} \beta^i \EE_z R_1 \cdots R_{i-1} \EE_{z_{i-1}} R_1    \\
    &\leq \sum_{i=0}^{n-1} \beta^i \EE_z R_1 \cdots R_{i-1} K_0 
    \leq \cdots \leq 
     \sum_{i=0}^{n-1} \beta^i K_0^i =: K_1 < \infty
\end{align*}
for all $z \in \ZZ$. Taking supremum on both sides yields
\begin{equation}
\label{eq:n-step-bd}
    \sup_{z \in \ZZ} \EE_{z}
    \left( \sum_{i=0}^{n-1} \beta^{i} R_{1} \cdots R_{i} \right) \leq K_1 < \infty.
\end{equation}
Based on \eqref{eq:n-step-ctra} and \eqref{eq:n-step-bd}, we have
\begin{align*}
    \EE_{z}
    \left( 
    \sum_{i = 0}^{\infty} \beta^i \prod_{k=1}^{i} R_{k}
    \right)
    &= \sum_{k=0}^{\infty} \EE_z 
    \left( \sum_{i=0}^{n-1} \beta^{kn+i} R_1 \cdots R_{kn + i}  \right)    \\
    &= \sum_{k=0}^{\infty} \EE_z 
    \left[
        \beta^{kn} R_1 \cdots R_{kn} 
        \left( \sum_{i=0}^{n-1} \beta^{i} R_{kn+1} \cdots R_{kn + i}  \right)    
    \right]    \\
    &= \sum_{k=0}^{\infty} \EE_z 
    \left[
    \beta^{kn} R_1 \cdots R_{kn} 
    \EE_{z_{kn}}
    \left( \sum_{i=0}^{n-1} \beta^{i} R_{1} \cdots R_{i}  \right)    
    \right]    \\
    &\leq 
    \sum_{k=0}^{\infty} \EE_z \beta^{kn} R_1 \cdots R_{kn} K_1  
    \leq 
    \sum_{k=0}^{\infty} \theta^k K_1 := K_2 < \infty
\end{align*}
for all $z \in \ZZ$. Hence,
\begin{equation*}
    \sup_{z \in \ZZ} \EE_{z}
    \left( 
    \sum_{i = 0}^{\infty} \beta^i \prod_{k=1}^{i} R_{k}
    \right) 
    \leq K_2 < \infty.
\end{equation*}
Finally, assumption \ref{a:Y_sum} implies that
\begin{align*}
  M(a,z) 
   \leq K_2 a + K_2 \sum_{t=1}^{\infty} \beta^t \EE_z Y_t < \infty
\end{align*}
for all $(a,z) \in \SS_0$. This concludes the proof. 
\end{proof}

\begin{proof}[Proof of theorem \ref{t:opt_result}]
This result extends theorem~1 of \cite{benhabib2015wealth} and theorem~3.1 of \cite{li2014solving}. While the assumptions are weaker in our setting, the proof is similar and hence omitted.
\end{proof}

In the next, we aim to prove proposition \ref{pr:complete}. To that end, we 
define $\hH$ to be the set of functions $h \colon \SS_0 \rightarrow \RR$ that 
satisfies
\begin{enumerate}
    \item $h$ is continuous,
    \item $h$ is decreasing in the first argument, and
    \item $\exists K \in \RR$ such that $u'(a) \leq h(a,z) \leq u'(a) + K$
        for all $(a,z) \in \SS_0$.
\end{enumerate}
On $\hH$ we impose the distance
\begin{equation}
\label{eq:dinf_metric}
    d_{\infty}(h,g) 
      := \left\| h - g \right\|
      := \sup_{(a,z) \in \SS_0} \left| h(a,z) - g(a,z) \right|.
\end{equation}
While the elements of $\hH$ are not bounded, the function $d_{\infty}$ is a 
valid metric. Moreover, standard argument shows that $(\hH, d_{\infty})$ is 
a complete metric space.

\begin{proof}[Proof of proposition \ref{pr:complete}]
Standard argument shows that $\rho$ is a valid metric. To show completeness 
of $(\cC, \rho)$,  it suffices to show that $(\cC, \rho)$ and 
$(\hH, d_{\infty})$ are isometrically isomorphic.

To see that this is so, let $H$ be the map on $\cC$ defined by 
$Hc = u' \circ c$. It is easy to show that 
$H: \cC \rightarrow \hH$ and that it is a bijection. Moreover, for all 
$c,d \in \cC$, 
\begin{equation*}
    d_{\infty}(Hc, Hd) = \left\| Hc - Hd \right\| 
                       = \left\| u' \circ c - u' \circ d \right\|
                       = \rho(c,d).
\end{equation*}
Hence, $H$ is an isometry. The space $(\cC, \rho)$ is then complete, as 
claimed.
\end{proof}

\begin{proof}[Proof of proposition \ref{pr:suff_optpol}]
Let $c$ be a policy in $\cC$ satisfying \eqref{eq:foc}. That $c$ satisfies the first order optimality  conditions is immediate by definition. It remains to show that any asset path generated by $c$ satisfies the transversality condition \eqref{eq:tvc}. To see that this is so, observe that, by \eqref{eq:bd_uprime},
\begin{equation}
\label{eq:ineq_betu'ca}
    \EE_{a,z} \beta^t (u' \circ c) (a_t, z_t) a_t
    \leq \beta^t \EE_{a, z} u'(a_t) a_t + 
         \beta^t K \EE_{a, z} a_t .
\end{equation}
Regarding the first term on the right hand side of \eqref{eq:ineq_betu'ca}, 
fix $L > 0$ and observe that
\begin{align*}
   \EE_{a,z} u'(a_t) a_t 
   &=
     \EE_{a,z} u'(a_t) a_t \1 \{a_t \leq L\} + 
     \EE_{a,z} u'(a_t) a_t \1 \{a_t > L\}    \\
   & \leq 
     L \EE_{a,z} u'(a_t) + u'(L) \EE_{a,z} a_t
   \leq 
     L \EE_{z} u'(Y_t) + u'(L) \EE_{a,z} \tilde{a}_t,
\end{align*}
where $\tilde{a}_t$  is the maximal path defined in \eqref{eq:max_path}. We 
then have
\begin{equation}
\label{eq:ineq_betEu'a}
  \beta^t \EE_{a,z} u'(a_t) a_t 
   \leq L \beta^t \EE_{z} u'(Y_t) + u'(L) \beta^t \EE_{a,z} \tilde{a}_t.
\end{equation} 
Since $M := \sup_{z \in \ZZ} \EE_z u'(\hat{Y}) < \infty$ by assumption \ref{a:bd_sup_ereuprm}, the Markov property then implies that for all $z \in \ZZ$ and $t \geq 1$,
\begin{equation*}
  \EE_z u'\left( Y_t \right) 
  = \EE_z \EE_z \left[ u' \left( Y_t \right) \big| \fF_{t-1} \right]
  = \EE_z \EE_{z_{t-1}} u' (\hat{Y})
  \leq \EE_z M = M.
\end{equation*}
Hence,
$\lim_{t \to \infty} \beta^t \EE_z u' \left( Y_t \right) = 0$. Since in addition $\lim_{t \to \infty} \beta^t \EE_{a,z} \tilde{a}_t = 0$ by lemma \ref{lm:max_path}, \eqref{eq:ineq_betEu'a} then implies that
$\lim_{t \to \infty} \beta^t \EE_{a,z} u'(a_t) a_t = 0$. 

Moreover, the second term on the right hand side of \eqref{eq:ineq_betu'ca} 
is dominated by $\beta^t K \EE_{a,z} \tilde{a}_t$, and converges to zero by 
lemma \ref{lm:max_path}. We have thus shown that the term on the right hand side of \eqref{eq:ineq_betu'ca}
converges to zero. Hence, the transversality condition holds.
\end{proof}

\begin{proof}[Proof of proposition \ref{pr:welldef_T}]
Fix $c \in \cC$ and $(a,z) \in \SS_0$. Because $c \in \cC$, the map 
$\xi \mapsto \psi_c(\xi, a, z)$ is increasing. Since $\xi \mapsto u'(\xi)$
is strictly decreasing, the equation \eqref{eq:T_opr} can have at most one 
solution. Hence uniqueness holds.

Existence follows from the intermediate value theorem provided we can show 
that 
\begin{enumerate}
  \item[(a)] $\xi \mapsto \psi_c(\xi, a, z)$ is a continuous function,
  \item[(b)] $\exists \xi \in (0,a]$ such that 
      $u'(\xi) \geq \psi_c(\xi, a, z)$, and
  \item[(c)] $\exists \xi \in (0,a]$ such that 
      $u'(\xi) \leq \psi_c(\xi, a, z)$.
\end{enumerate}
For part (a), it suffices to show that 
$g(\xi) := 
  \EE_{z} \hat{R} 
          \left(u' \circ c \right) 
          \left[ \hat{R}(a - \xi) + \hat{Y}, \hat{z} \right]$
is continuous on $(0,a]$. To this end, fix $\xi \in (0,a]$ and 
$\xi_n \rightarrow \xi$. By \eqref{eq:bd_uprime} we have
\begin{align}
\label{eq:uppbd_ruprmc}
    \hat{R} \left( u' \circ c \right) 
        \left[ \hat{R} \left( a - \xi \right) + \hat{Y}, \hat{z} \right] 
    \leq 
        \hat{R} \left( u' \circ c \right) ( \hat{Y}, \hat{z} )
    \leq 
        \hat{R} u'( \hat{Y}) + \hat{R} K.
\end{align}
The last term is integrable by assumption \ref{a:bd_sup_ereuprm}. Hence the dominated convergence theorem applies. From this fact and the continuity of $c$, we obtain $g(\xi_n) \rightarrow g(\xi)$. Hence, $\xi \mapsto \psi_c(\xi, a, z)$ is continuous.

Part (b) clearly holds, since $u'(\xi) \rightarrow \infty$ as 
$\xi \rightarrow 0$ and $\xi \mapsto \psi_c(\xi, a, z)$ is increasing and 
always finite (since it is continuous as shown in the previous paragraph). 
Part (c) is also trivial (just set $\xi = a$).
\end{proof}

\begin{proof}[Proof of proposition \ref{pr:self_map}]
Fix $c \in \cC$. With slight abuse of notation, we denote
\begin{equation*} 
g \left( \xi, a, z \right) := 
  \EE_{z} \hat{R} 
          \left( u' \circ c \right) 
          \left[ \hat{R} \left( a - \xi \right) + \hat{Y}, 
                 \, \hat{z} 
          \right].
\end{equation*}
\textbf{Step~1.} We show that $Tc$ is continuous. To apply a standard 
fixed point parametric continuity result such as theorem~B.1.4 of 
\cite{stachurski2009economic}, we first show that $\psi_c$ is jointly 
continuous on the set $G$ defined in \eqref{eq:dom_T_opr}. This will be true
if $g$ is jointly continuous on $G$. For any $\{ (\xi_n, a_n, z_n) \}$ and 
$(\xi, a, z)$ in $G$ with $(\xi_n, a_n, z_n) \rightarrow (\xi, a, z)$, we 
need to show that $g(\xi_n, a_n, z_n) \rightarrow g(\xi, a, z)$. To that 
end, we define
\begin{align*}
 h_1 ( \xi, a, \hat{z}, \hat{\zeta}, \hat{\eta} ), \,
  h_2 ( \xi, a, \hat{z}, \hat{\zeta}, \hat{\eta} )   
  := \hat{R} 
     [ u' (\hat{Y}) + K ]
       \pm 
     \hat{R} \left( u' \circ c \right)
     [ \hat{R} \left( a - \xi \right) + \hat{Y}, \hat{z} ],
\end{align*}
where $\hat{R} := R ( \hat{z}, \hat{\zeta} )$ and 
$\hat{Y} := Y ( \hat{z}, \hat{\eta} )$ as defined in 
\eqref{eq:RY_func}. Then $h_1$ and $h_2$ are continuous in $(\xi, a, \hat{z})$ by 
the continuity of $c$ and assumption \ref{a:conti_ereuprm}, and they are nonnegative since \eqref{eq:uppbd_ruprmc} implies that
$0 \leq 
      \hat{R} \left( u' \circ c \right) 
      [ \hat{R} \left( a - \xi \right) + \hat{Y}, \hat{z} ]   
   \leq \hat{R} [ u' (\hat{Y}) + K ]$.

Moreover, since the stochastic kernel $P$ is Feller, the product measure 
satisfies\footnote{Here $\stackrel{w}{\to}$ denotes weak convergence, i.e.,  
    for all bounded continuous function $f$, we have
    \begin{equation*}
      \int 
          f(\hat{z}, \hat{\zeta}, \hat{\eta}) 
      P(z_n, \diff \hat{z}) \nu(\diff \hat{\zeta}) \mu (\diff \hat{\eta})
      \to
       \int 
          f(\hat{z}, \hat{\zeta}, \hat{\eta}) 
      P(z, \diff \hat{z}) \nu(\diff \hat{\zeta}) \mu (\diff \hat{\eta}).
    \end{equation*}
    The formal definition of weak convergence is provided in section 
    \ref{ss:gs_iid}. 
    }
\begin{equation*}
  P(z_n, \cdot) \otimes \nu \otimes \mu 
  \stackrel{w}{\longrightarrow} 
  P(z, \cdot) \otimes \nu \otimes \mu.
\end{equation*}
Based on the generalized Fatou's lemma of \cite{feinberg2014fatou}
(theorem~1.1), 
\begin{align*}
\liminf_{n \rightarrow \infty} &
  \int 
      h_i ( \xi_n, a_n, \hat{z}, \hat{\zeta}, \hat{\eta} )
  P( z_n, \diff \hat{z}) \nu(\diff \hat{\zeta}) \mu(\diff \hat{\eta})   \\
&\geq
  \int 
      h_i ( \xi, a, \hat{z}, \hat{\zeta}, \hat{\eta} )
  P( z, \diff \hat{z}) \nu(\diff \hat{\zeta}) \mu(\diff \hat{\eta}).  
\end{align*}
Since $z \mapsto \EE_z \hat{R} \, [u'(\hat{Y}) + K ]$ is continuous by 
assumption \ref{a:conti_ereuprm}, this implies that
\begin{align*}
 \liminf_{n \rightarrow \infty} 
  \left(
      \pm 
      \EE_{z_n}        
         \hat{R} 
         \left( u' \circ c \right)
         \left[ \hat{R} \left(a_n - \xi_n \right) + \hat{Y}, \hat{z} \right]
  \right)    
 \geq  
  \left(
      \pm 
      \EE_{z}   
        \hat{R} 
        \left(u' \circ c \right)
        \left[ \hat{R} \left(a - \xi \right) + \hat{Y}, \hat{z} \right]
  \right).  
\end{align*}
The function $g$ is then continuous since the above inequality is equivalent 
to
\begin{align*}
\liminf_{n \rightarrow \infty} g(\xi_n, a_n, z_n)   
  \geq g(\xi, a, z)   
  \geq   
  \limsup_{n \rightarrow \infty} g(\xi_n, a_n, z_n).
\end{align*}
Hence, $\psi_c$ is continuous on $G$, as was to be shown.
Moreover, since $\xi \mapsto \psi_c(\xi, a, z)$ takes values in the closed 
interval 
\begin{equation*}
  I(a,z) := \left[ 
                u'(a), 
                u'(a) + \EE_z \hat{R} \left( u'(\hat{Y}) + K \right)
            \right],
\end{equation*}
the correspondence $(a, z) \mapsto I(a,z)$ is nonempty, compact-valued and 
continuous. By theorem~B.1.4 of \cite{stachurski2009economic}, 
$(a,z) \mapsto [u' \circ (Tc)] (a,z)$ is continuous. $Tc$ is then continuous 
on $\SS_0$ since $u'$ is continuous.

\textbf{Step 2.} We show that $Tc$ is increasing in $a$. Suppose that for 
some $z \in \ZZ$ and $a_1, a_2 \in (0, \infty)$ with $a_1 < a_2$, we have
$\xi_1 := Tc (a_1,z) > Tc (a_2,z) =: \xi_2$. Since $c$ is increasing in $a$ 
by assumption, $\psi_c$ is increasing in $\xi$ and decreasing in $a$. Then 
$u'(\xi_1) < u'(\xi_2) 
 = \psi_c(\xi_2, a_2, z) 
 \leq \psi_c(\xi_1, a_1, z) = u'(\xi_1)$. This is a contradiction.
 
\textbf{Step 3.} We have shown in proposition \ref{pr:welldef_T} that 
$Tc(a,z) \in (0,a]$ for all $(a,z) \in \SS_0$.

\textbf{Step 4.} We show that $\| u' \circ (Tc) - u' \| < \infty$. Since
$u'[Tc(a,z)] \geq u'(a)$, we have
\begin{align*}
    &\left| u'[Tc(a,z)] - u'(a) \right| 
    = u'[Tc(a,z)] - u'(a)    \\
    & \leq
      \EE_{z} \hat{R} 
              \left( u' \circ c \right)  
              \left( \hat{R} \left[ a - Tc(a,z) \right] + \hat{Y},
                     \, \hat{z})
              \right)    
    \leq 
      \EE_{z} \hat{R} \left[ u'(\hat{Y}) + K \right].
\end{align*}
for all $(a,z) \in \SS_0$. Assumption \ref{a:bd_sup_ereuprm} then implies that
\begin{align*}
\left\| u' \circ (Tc) - u' \right\| 
\leq 
    \sup_{z \in \ZZ} \EE_{z} \hat{R} u'(\hat{Y}) + 
    K \left( \sup_{z \in \ZZ} \EE_z \hat{R} \right) < \infty.
\end{align*}
This concludes the proof.
\end{proof}

In the rest of this section, we aim to prove theorem \ref{t:ctra_T}. Recall
$\hH$ defined above.
Given $h \in \hH$, let $\tilde{T} h$ be the function mapping $(a,z) \in \SS_0$ 
into the $\kappa$ that solves
\begin{equation}
\label{eq:T_hat}
    \kappa = 
    \max \left\{ 
            \beta \EE_{z} \hat{R}  \,
                  h \left( 
                        \hat{R} \left[ 
                                    a - \left(u' \right)^{-1}(\kappa) 
                                \right]  
                        + \hat{Y}, 
                        \, \hat{z} 
                    \right),
            \, u'(a)
         \right\}.
\end{equation}
The next lemma implies that $\tilde{T}$ is a well-defined self-map on $\hH$, 
as well as topologically conjugate to $T$ under the bijection 
$H: \cC \rightarrow \hH$ defined by $Hc := u' \circ c$.

\begin{lemma}
\label{lm:conjug}
  The operator $\tilde{T} \colon \hH \to \hH$ and satisfies $\tilde{T} H = H T $ on $\cC$.
\end{lemma}

\begin{proof}[Proof of lemma \ref{lm:conjug}]
Pick any $c \in \cC$ and $(a,z) \in \SS_0$. Let $\xi := Tc(a,z)$, then $\xi$ solves
\begin{equation}
\label{eq:Tc_eq}
    u'(\xi) = 
     \max \left\{
             \beta \EE_{z} \hat{R} 
                   \left(u' \circ c \right) 
                   \left[ 
                       \hat{R} \left(a - \xi \right) + \hat{Y},
                       \, \hat{z}
                   \right],
             \, u'(a)
          \right\}.
\end{equation}
We need to show that $HTc$ and $\tilde{T} Hc$ evaluate to the same number at 
$(a,z)$. In other words, we need to show that $u'(\xi)$ is the solution to
\begin{equation*}
  \kappa = 
    \max \left\{
            \beta \EE_{z} \hat{R}
                  \left( u' \circ c \right) 
                     \left( 
                        \hat{R} \left[ 
                                   a - \left(u' \right)^{-1} (\kappa)
                                \right] 
                        + \hat{Y},
                        \, \hat{z}
                     \right),
            \, u'(a)
        \right\}.
\end{equation*}
But this is immediate from \eqref{eq:Tc_eq}. Hence, we have shown that
$\tilde{T} H = H T$ on $\cC$. Since $H \colon \cC \to \hH$ is a bijection,
we have $\tilde{T} = HT H^{-1}$. Since in addition $T \colon \cC \to \cC$ by 
proposition \ref{pr:self_map}, we have $\tilde{T} \colon \hH \to \hH$. This 
concludes the proof.
\end{proof}

\begin{lemma}
\label{lm:monot}
$\tilde{T}$ is order preserving on $\hH$. That is, 
$\tilde{T} h_1 \leq \tilde{T} h_2$ for all $h_1, h_2 \in \hH$ with $h_1 \leq h_2$.
\end{lemma}

\begin{proof}[Proof of lemma \ref{lm:monot}]
Let $h_1, h_2$ be functions in $\hH$ with $h_1 \leq h_2$. Suppose to the 
contrary that there exists $(a,z) \in \SS_0$ such that 
$\kappa_1 := \tilde{T} h_1 (a,z) > \tilde{T} h_2 (a,z) =: \kappa_2$. Since functions in $\hH$ are decreasing in the first argument, we have
\begin{align*}
  \kappa_1 &= 
      \max \left\{
              \beta \EE_{z} \hat{R} \, 
                    h_1 \left( \hat{R} 
                               \left[ a - (u')^{-1}(\kappa_1) \right]
                               + \hat{Y},
                               \, \hat{z} 
                        \right),
              \, u'(a)
           \right\}      \\
     & \leq 
      \max \left\{
              \beta \EE_{z} \hat{R} \,
                    h_2 \left( \hat{R} 
                               \left[ a - (u')^{-1}(\kappa_1) \right]
                               + \hat{Y},
                               \, \hat{z} 
                        \right),
              \, u'(a)
           \right\}      \\
     & \leq 
      \max \left\{
              \beta \EE_{z} \hat{R} \, 
                    h_2 \left( \hat{R} 
                               \left[ a - (u')^{-1}(\kappa_2) \right]
                               + \hat{Y},
                               \, \hat{z} 
                        \right),
              \, u'(a)
           \right\}
     = \kappa_2.
\end{align*}
This is a contradiction. Hence, $\tilde{T}$ is order preserving.
\end{proof}

\begin{lemma}
\label{lm:ctra_That}
$\tilde{T}^n$ is a contraction mapping on $(\hH, d_{\infty})$ with modulus 
$\theta$.
\end{lemma}

\begin{proof}[Proof of lemma \ref{lm:ctra_That}]
Since $\tilde{T}$ is order preserving and $\hH$ is closed under the addition of nonnegative constants, based on \cite{blackwell1965discounted}, it remains to verify: for $n$ and $\theta$ given by assumption \ref{a:ctra_coef},
\begin{equation*}
  \tilde{T}^n (h+\gamma) \leq \tilde{T}^n h + \theta \gamma  \;
   \text{ for all } h \in \hH 
   \text{ and } \gamma \geq 0.
\end{equation*}
To that end, by assumption \ref{a:ctra_coef}, it suffices to show that for all $k \in \NN$ and $(a,z) \in \SS_0$,
\begin{equation}
\label{eq:That_k}
    \tilde{T}^k (h+ \gamma) (a,z)
    \leq \tilde{T}^k h(a,z) + \gamma \beta^k \EE_z R_1 \cdots R_k.
\end{equation}
Fix $h \in \hH$, $\gamma \geq 0$, and let $h_{\gamma} (a,z) := h(a, z) + \gamma$. By the definition of $\tilde{T}$, we have
\begin{align*}
\tilde{T} h_\gamma (a,z) 
&= 
   \max \left\{ 
           \beta \EE_{z} \hat{R} \, 
                 h_{\gamma} 
                   \left( \hat{R} 
                          \left[ 
                              a - (u')^{-1}(\tilde{T} h_{\gamma})(a,z) 
                          \right] 
                          + \hat{Y}, 
                          \hat{z} 
                   \right),
            u'(a)
        \right\}    \\
& \leq 
   \max \left\{ 
           \beta \EE_{z} \hat{R} \, 
                 h \left( \hat{R} 
                          \left[ a - 
                                (u')^{-1} (\tilde{T} h_{\gamma})(a,z) 
                          \right] 
                          + \hat{Y}, 
                          \hat{z} 
                   \right),
            u'(a)
        \right\}
    + \gamma \beta \EE_z R_1    \\
& \leq
   \max \left\{ 
           \beta \EE_{z} \hat{R} \,
                 h \left( \hat{R} 
                          \left[ a - 
                                (u')^{-1} (\tilde{T} h)(a,z) 
                          \right] 
                          + \hat{Y}, 
                          \hat{z} 
                   \right),
            u'(a)
        \right\}
    + \gamma \beta \EE_z R_1.
\end{align*}
Here, the first inequality is elementary and the second is due to the fact 
that $h \leq h_\gamma$ and $\tilde{T}$ is order preserving. Hence, $\tilde{T} (h+ \gamma) (a,z ) \leq \tilde{T} h (a,z) + \gamma \beta \EE_z R_1$ and \eqref{eq:That_k} holds for $k = 1$. Suppose that \eqref{eq:That_k} holds for arbitrary $k$. It remains to show that \eqref{eq:That_k} holds for $k+1$. Define
\begin{equation*}
    f(z) := \gamma \beta^k \EE_z R_1 \cdots R_k.
\end{equation*}
By the induction hypothesis, the monotonicity of $\tilde{T}$ and the Markov property,
\begin{align*}
    \tilde{T}^{k+1} h_\gamma (a,z)
    &= \max \left\{ 
        \beta \EE_{z} \hat{R} \, (\tilde{T}^k h_{\gamma})
        \left( \hat{R} 
        \left[ a - (u')^{-1}(\tilde{T}^{k+1} h_{\gamma})(a,z) \right] 
        + \hat{Y}, \hat{z} 
       \right),
       u'(a)
   \right\}    \\
   & \leq \max \left\{
       \beta \EE_z \hat{R} 
       \left(\tilde{T}^k h + f \right) \left( \hat{R} 
       \left[ a - (u')^{-1}(\tilde{T}^{k+1} h_{\gamma})(a,z) \right] 
       + \hat{Y}, \hat{z} 
       \right),
       u'(a)
   \right\}    \\
   & \leq \max \left\{
   \beta \EE_z \hat{R} (\tilde{T}^k h) \left( \hat{R} 
   \left[ a - (u')^{-1}(\tilde{T}^{k+1} h_{\gamma})(a,z) \right] 
   + \hat{Y}, \hat{z} \right), u'(a) \right\}    \\
   & \quad + \beta \EE_{z} R_1 f(z_1)    \\
   & \leq  \max \left\{
   \beta \EE_z \hat{R} (\tilde{T}^k h) \left( \hat{R} 
   \left[ a - (u')^{-1}(\tilde{T}^{k+1} h)(a,z) \right] 
   + \hat{Y}, \hat{z} \right), u'(a) \right\}     \\ 
   &\quad + \gamma \beta^{k+1} \EE_{z} R_1 \EE_{z_1} R_1 \cdots R_{k}    \\
   &= \tilde{T}^{k+1} h(a,z) + \gamma \beta^{k+1} \EE_z R_1 \cdots R_{k+1}.
\end{align*} 
Hence, \eqref{eq:That_k} is verified by induction. This concludes the proof.
\end{proof}

With the results established above, we are now ready to prove theorem \ref{t:ctra_T}.

\begin{proof}[Proof of theorem \ref{t:ctra_T}]
In view of propositions \ref{pr:complete} and \ref{pr:suff_optpol}, to 
establish all the claims in theorem \ref{t:ctra_T}, we need only show that
\begin{equation*}
  \rho(T^n c, T^n d) \leq \theta \rho(c,d) \quad \text{for all } \, c,d \in \cC.
\end{equation*}
To this end, pick any $c,d \in \cC$. Note that the topological conjugacy result established in lemma \ref{lm:conjug} implies that $\tilde{T} = H T H^{-1}$. Hence,
\begin{equation*}
    \tilde{T}^n = (H T H^{-1}) \cdots (H T H^{-1}) = H T^n H^{-1} 
    \quad \text{and} \quad
    \tilde{T}^n H = H T^n.
\end{equation*}
By the definition of $\rho$ and the contraction property established in lemma \ref{lm:ctra_That},
\begin{equation*}
  \rho(T^n c, T^n d) = d_{\infty}(H T^n c, H T^n d) 
               = d_{\infty}(\tilde{T}^n Hc, \tilde{T}^n Hd)
               \leq \theta d_{\infty}(Hc, Hd).
\end{equation*}
The right hand side is just $\theta \rho(c,d)$, which completes the proof.
\end{proof}

\section{Appendix B: Proof of Section \ref{s:sto_stability} Results}
\label{Appendix_B}

Before working into the results of each subsection, we prove a general lemma that is frequently used in later sections. Recall that, for all $c \in \cC$, the value $\xi(a,z) := Tc(a,z)$ solves
\begin{equation}
\label{eq:T_opr_general}
  \left( u' \circ \xi \right)(a,z) =
    \max \left\{
            \beta \EE_{z} \hat{R} 
                   \left(u' \circ c \right) 
                         \left(
                              \hat{R} \left[a - \xi(a,z) \right] + \hat{Y},
                              \, \hat{z}
                         \right),
            \, u'(a)
         \right\}.
\end{equation}
Let $c^* \in \cC$ denote the optimal policy. For each $z \in \ZZ$ and $c \in \cC$, define
\begin{equation}
\label{eq:a_bar}
  \bar{a}_c (z) := 
      \left(u' \right)^{-1}
      \left[ 
          \beta \EE_z \hat{R} 
                   \left(u' \circ c \right) 
                   ( \hat{Y}, \, \hat{z} ) 
      \right]    
  \quad \text{and} \quad
  \bar{a}(z) := \bar{a}_{c^*} (z).
\end{equation}
The next result implies that the borrowing constraint binds if and only if wealth is below a certain threshold level.

\begin{lemma}
\label{lm:binding}
  For all $c \in \cC$, $Tc(a,z) = a$ if and only if $a \leq \bar{a}_c (z)$. In particular,
  $c^*(a,z) = a$ if and only if $a \leq \bar{a}(z)$.
\end{lemma}

\begin{proof}[Proof of lemma \ref{lm:binding}]
Let $a \leq \bar{a}_c (z)$. We claim that $\xi(a,z) = a$. Suppose to
the contrary that $\xi(a,z) < a$. Then $(u' \circ \xi)(a,z) > u'(a)$. In view of \eqref{eq:T_opr_general}, we have
\begin{align*}
  u'(a) < \beta \EE_z \hat{R} 
                 \left(u' \circ c \right) 
                     \left(
                          \hat{R} \left[a - \xi(a,z) \right] + \hat{Y},
                          \, \hat{z}
                     \right)    
         \leq
           \beta \EE_z \hat{R} 
                 \left(u' \circ c \right) (\hat{Y}, \hat{z})
        = u'[\bar{a}_c (z)].
\end{align*}
From this we get $a > \bar{a}_c (z)$, which is a contradiction. Hence,
$\xi(a,z) = a$.

On the other hand, if $\xi(a,z) = a$, then $(u' \circ \xi)(a,z) = u'(a)$. By 
\eqref{eq:T_opr_general}, we have
\begin{align*}
  u'(a) 
  \geq 
       \beta \EE_z \hat{R}
                   \left( u' \circ c \right) (\hat{Y}, \, \hat{z})
  = u'[\bar{a}_c (z)].                    
\end{align*} 
Hence, $a \leq \bar{a}_c (z)$. The first claim is verified. The second claim follows immediately from the first claim and the fact that $c^*$ is the unique fixed point of  $T$ in $\cC$.
\end{proof}

Given $c \in \cC$, lemma \ref{lm:binding} implies that $\xi(a,z) := Tc (a,z) = a$ for 
$a \leq \bar{a}_c (z)$, and that for $a > \bar{a}_c (z)$, $\xi(a,z)$ solves
\begin{equation*}
  (u' \circ \xi)(a,z)
  = \beta \EE_z \hat{R}
               \left(u' \circ c \right) 
               \left(
                    \hat{R} \left[a - \xi(a,z) \right] + \hat{Y},
                    \, \hat{z}
               \right).
\end{equation*}

\subsection{Proof of section \ref{ss:exist_stat} results}

Our first goal is to prove proposition \ref{pr:opt_pol_bd_frac}. To that end,
recall $\alpha$ given by assumption \ref{a:suff_bd_in_prob}, and define the subspace $\cC_1$ as 
\begin{equation}
\label{eq::cC1}
    \cC_1 := 
    \left\{ 
        c \in \cC: \frac{c(a,z)}{a} \geq \alpha 
        \quad \text{for all } (a,z) \in \SS_0
    \right\}.
\end{equation}

\begin{lemma}
\label{lm:cC1}
    $\cC_1$ is a closed subset of $\cC$, and $T c \in \cC_1$ for all $c \in \cC_1$.
\end{lemma}

\begin{proof}[Proof of lemma \ref{lm:cC1}]
To see that $\cC_1$ is closed, for a given sequence $\{ c_n \}$ in $\cC_1$ and 
$c \in \cC$ with $\rho( c_n , c) \rightarrow 0$, we need to 
verify that $c \in \cC_1$. This obviously holds since $c_n(a,z) /a \geq \alpha$ for all $n$ and $(a,z) \in \SS_0$, and, on the other hand, $\rho( c_n , c) \rightarrow 0$ implies that $c_n (a,z) \to c(a,z)$ for all $(a,z) \in \SS_0$. 

We next show that $T$ is a self-map on $\cC_1$. Fix $c \in \cC_1$. We have $Tc \in \cC$ since $T$ is a self-map on $\cC$. It remains to show that $\xi := Tc$ satisfies 
$\xi (a,z) \geq \alpha a$ for all $(a,z) \in \SS_0$. Suppose to the contrary that 
$\xi(a,z) < \alpha a$ for some $(a,z) \in \SS_0$. Then 
\begin{equation*}
u'(\alpha a) < (u' \circ \xi)(a, z)    
 = \max 
       \left\{ 
             \beta \EE_z \hat{R}    
              \left( u' \circ c \right) 
              \left( 
                   \hat{R} \left[ a - \xi(a,z) \right] + \hat{Y},
                   \, \hat{z}
              \right),
              u'(a)
      \right\}. 
\end{equation*}
Since $u'(\alpha a) > u'(a)$ and $c \in \cC_1$, this implies that
\begin{align*}
 u'(\alpha a )
   &< \beta \EE_z \hat{R}    
              \left( u' \circ c \right) 
              \left( 
                   \hat{R} \left[ a - \xi(a,z) \right] + \hat{Y},
                   \, \hat{z}
              \right)     \\
   &\leq \beta \EE_z \hat{R}    
                   u' \left( 
                          \alpha \hat{R} \left[ a - \xi(a,z) \right] 
                          + \alpha \hat{Y}
                      \right)    \\
& \leq \beta \EE_z \hat{R}    
                   u' \left[
                          \alpha \hat{R} (1 - \alpha) a + \alpha \hat{Y}
                      \right]    
  \leq \beta \EE_z \hat{R}  \,  
                   u' \left[ 
                          \hat{R} (1 - \alpha) (\alpha a)
                      \right].
\end{align*}
This is a contradicted with condition (1) of assumption \ref{a:suff_bd_in_prob} 
since $(\alpha a, z) \in \SS_0$. Hence, $\xi(a,z) / a \geq \alpha$ for all 
$(a,z) \in \SS_0$ and we conclude that $Tc \in \cC_1$.
\end{proof}

With this result, we are now ready to prove proposition \ref{pr:opt_pol_bd_frac}.

\begin{proof}[Proof of proposition \ref{pr:opt_pol_bd_frac}]
Since the claim obviously holds when $a=0$, it remains to verify that this claim holds on $\SS_0$. We have shown in theorem \ref{t:ctra_T} that $T$ is a contraction mapping on the complete metric space $(\cC, \rho)$, with unique fixed point $c^*$. Since in addition $\cC_1$ is a closed subset of $\cC$ and $T \cC_1 \subset \cC_1$ by lemma \ref{lm:cC1}, we know that $c^* \in \cC_1$. In summary, we have $c^*(a, z) \geq \alpha a$ for all $(a,z) \in \SS$.  
\end{proof}

Our next goal is to prove theorem \ref{t:sta_exist}. To that end, recall the integer $n$ given by the second condition of assumption \ref{a:suff_bd_in_prob}.

\begin{lemma}
\label{lm:bd_in_prob_at}
$\sup_{t \geq 0} \EE_{a,z} \, a_t  < \infty$
for all $(a,z) \in \SS$.
\end{lemma}

\begin{proof}[Proof of lemma \ref{lm:bd_in_prob_at}]
Since $c^*(0,z) = 0$, proposition \ref{pr:opt_pol_bd_frac} implies that 
$c^*(a,z) \geq \alpha a$ for all $(a,z) \in \SS$. For all $t \geq 1$, we have $t = kn + j$ in general, where $k \in \{0\} \cup \NN$ and $j \in \{0,1, \cdots, n-1\}$. Using these facts and \eqref{eq:trans_at}, we have:
\begin{align*}
a_t &= R_t(a_{t-1} - c_{t-1}) + Y_t 
    \leq (1 - \alpha) R_t a_{t-1} + Y_t \leq \cdots   \\
    & \leq (1 - \alpha)^t R_t \cdots R_1 a + (1 - \alpha)^{t-1} R_t \cdots R_2 Y_1
    + \cdots + (1 - \alpha) R_t Y_{t-1} + Y_t    \\
    &= (1 - \alpha)^{kn+j} R_{kn +j} \cdots R_1 a + 
    \sum_{\ell=1}^{j} (1 - \alpha)^{kn + j - \ell} R_{kn+j} \cdots R_{\ell+1} Y_{\ell}    \\
    & \quad + \sum_{m=1}^{k} \sum_{\ell=1}^{n} (1 - \alpha)^{mn - \ell}
    R_{kn + j} \cdots R_{(k-m)n + j + \ell + 1} Y_{(k-m)n + j + \ell}
\end{align*}
with probability one. Hence,
\begin{align*}
    \EE_{a,z} a_t 
    &\leq (1 - \alpha)^t \EE_z R_t \cdots R_1 a + 
    \sum_{\ell = 1}^{t} (1 - \alpha)^{t-\ell} \EE_z R_t \cdots R_{\ell+1} Y_{\ell}    \\
    &= (1 - \alpha)^{kn+j} \EE_z R_{kn +j} \cdots R_1 a + 
    \sum_{\ell=1}^{j} (1 - \alpha)^{kn + j - \ell} \EE_z R_{kn+j} \cdots R_{\ell+1} Y_{\ell}    \\
    & \quad + \sum_{m=1}^{k} \sum_{\ell=1}^{n} (1 - \alpha)^{mn - \ell}
    \EE_z R_{kn + j} \cdots R_{(k-m)n + j + \ell + 1} Y_{(k-m)n + j + \ell}     
\end{align*}
for all $(a,z) \in \SS$. Define 
\begin{equation*}
    \gamma := (1 - \alpha)^n \sup_{z \in \ZZ} \EE_z R_1 \cdots R_n
    \quad \text{and} \quad
    M:= \max_{1 \leq \ell \leq n} 
        \left[
            (1 - \alpha)^\ell \sup_{z \in \ZZ} \EE_z R_\ell \cdots R_1
        \right].
\end{equation*}
Note that $\gamma < 1$ by assumption \ref{a:suff_bd_in_prob}-(2) and $M < \infty$ by assumption \ref{a:bd_sup_ereuprm} and the Markov property. Moreover, $M' := \sup_{t \geq 1} \EE_z Y_t < \infty$ by assumption \ref{a:bd_in_prob_Yt}. The Markov property then implies that for all $(a,z) \in \SS$ and $t \geq 0$,
\begin{align*}
	\EE_{a,z} a_t 
	&\leq \gamma^k (1 - \alpha)^{j} \EE_z R_{j} \cdots R_1 a + 
	\gamma^k \sum_{\ell=1}^{j} (1 - \alpha)^{j - \ell} \EE_z R_{j} \cdots R_{\ell+1} Y_{\ell}    \\
	& \quad + \sum_{m=0}^{k-1} \gamma^m \sum_{\ell=1}^{n} (1 - \alpha)^{n - \ell}
	\EE_z R_{(k-m)n + j} \cdots R_{(k-m-1)n + j + \ell + 1} Y_{(k-m)n + j + \ell}     \\
	&\leq \gamma^k M a + \gamma^k M \sum_{\ell=1}^{j} \EE_z Y_\ell +
	\sum_{m=0}^{k-1} \gamma^m M \sum_{\ell=1}^{n} \EE_z Y_{(k-m-1)n + j + \ell}  \\
	&\leq Ma + MM'n + \sum_{m=0}^{\infty} \gamma^m MM' n < \infty.
\end{align*}
Hence, $\sup_{t \geq 0} \EE_{a,z} \, a_t < \infty$ for all $(a,z) \in \SS$, as was claimed.
\end{proof}

A function $w^* \colon \SS \to \RR_+$ is called \emph{norm-like} if all 
its sublevel sets (i.e., sets of the form 
$\{s \in \SS \colon w(s) \leq b \}, b \in \RR_+$) are precompact in $\SS$ 
(i.e., any sequence in a given sublevel set has a subsequence that converges 
to a point of $\SS$).

\begin{proof}[Proof of theorem \ref{t:sta_exist}]
Based on lemma~D.5.3 of \cite{meyn2009markov}, a stochastic kernel $Q$ is 
bounded in probability if and only if for all $s \in \SS$, there exists a 
norm-like function $w_s^* \colon \SS \to \RR_+$ such that the $(Q,s)$-Markov process $\{s_t\}_{t \geq 0}$ satisfies $\limsup_{t \to \infty} \EE_s \left[ w_s^*(s_t) \right] < \infty$. 

Fix $(a,z) \in \SS$. Since $P$ is bounded in probability by assumption 
\ref{a:z_bdd_in_prob}, there exists a norm-like function 
$w \colon \ZZ \to \RR_+$ such that 
$\limsup_{t \to \infty} \EE_z w (z_t) < \infty$.
Then $w^* \colon \SS \rightarrow \RR_+$ defined by 
$w^*(a_0,z_0) := a_0 + w (z_0)$ is 
a norm-like function on $\SS$. The stochastic kernel $Q$ is then bounded in 
probability since lemma \ref{lm:bd_in_prob_at} implies that
\begin{equation*}
  \limsup_{t \to \infty} \EE_{a,z} \, w^*(a_t, z_t)
  \leq \sup_{t \geq 0} \EE_{a,z} \, a_t +
       \limsup_{t \to \infty} \EE_z \, w(z_t) 
  < \infty.
\end{equation*}
Regarding existence of stationary distribution, since $c^*$ is continuous and 
assumption \ref{a:conti_ereuprm} holds, and we have shown in the proof 
of proposition \ref{pr:self_map} that 
\begin{equation*}
  P(z_n, \cdot) \otimes \nu \otimes \mu 
  \stackrel{w}{\longrightarrow}
  P(z, \cdot) \otimes \nu \otimes \mu
\end{equation*}
whenever $z_n \to z$, a simple application of the generalized Fatou's lemma
of \cite{feinberg2014fatou} (theorem~1.1) as in the proof of proposition 
\ref{pr:self_map} shows that the stochastic kernel $Q$ is Feller. 
Since in addition $Q$ is bounded in probability, based on the 
Krylov-Bogolubov theorem (see, e.g., \cite{meyn2009markov}, 
proposition~12.1.3 and lemma~D.5.3), $Q$ admits at least one stationary 
distribution.
\end{proof}

\subsection{Proof of section \ref{ss:further_prop} results}

We start by proving example \ref{eg:concave}.

\begin{proof}[Proof of example \ref{eg:concave}]
For each $c$ in $\cC$ concave in the first argument, let 
$h_c (x, \hat{\omega}) 
   := c ( \hat{R} x + \hat{Y}, \hat{z} )$, 
where $\hat{\omega} := ( \hat{R}, \hat{Y}, \hat{z} )$. 
Then $x \mapsto h_c (x, \hat{\omega})$ is concave. Since 
$u'(c) = c^{-\gamma}$, we have 
\begin{align*}
   &\left[ 
       \beta 
       \EE_z \hat{R} \, 
             h_c(\alpha x_1 + (1 - \alpha) x_2, \hat{\omega})^{-\gamma} 
    \right]^{-\frac{1}{\gamma}}   
   \geq  
    \left[
        \beta \EE_z \hat{R} 
              \left[
                  \alpha h_c (x_1, \hat{\omega}) +
                  (1 - \alpha) h_c (x_2, \hat{\omega}) 
              \right]^{-\gamma}
    \right]^{-\frac{1}{\gamma}}   \\
  &= \beta^{-\frac{1}{\gamma}}
    \left(
        \EE_z \left[
                  \alpha 
                  \hat{R}^{-\frac{1}{\gamma}} h_c(x_1, \hat{\omega})
                     +
                  (1 - \alpha)
                  \hat{R}^{-\frac{1}{\gamma}} h_c(x_2, \hat{\omega})
              \right]^{-\gamma}
    \right)^{-\frac{1}{\gamma}}    \\
  &\geq \beta^{-\frac{1}{\gamma}}
    \left[
      \left(  
          \EE_z \left[
                \alpha
                \hat{R}^{-\frac{1}{\gamma}} 
                h_c(x_1, \hat{\omega})
                \right]^{-\gamma} 
      \right)^{-\frac{1}{\gamma}}  + 
       \left(  
          \EE_z \left[
                    (1 - \alpha)
                    \hat{R}^{-\frac{1}{\gamma}} 
                    h_c(x_2, \hat{\omega})
                \right]^{-\gamma} 
       \right)^{-\frac{1}{\gamma}} 
    \right]  \\
  &= 
    \alpha
      \left[ 
          \beta 
          \EE_z \hat{R} \,
                h_c(x_1, \hat{\omega})^{-\gamma} 
      \right]^{-\frac{1}{\gamma}}  + 
   (1 - \alpha)
      \left[ 
          \beta 
          \EE_z \hat{R} \, 
                h_c(x_2, \hat{\omega})^{-\gamma} 
      \right]^{-\frac{1}{\gamma}},
\end{align*}
where the second inequality is due to the generalized Minkowski's inequality 
(see, e.g., \cite{hardy1952inequalities}, page~146, theorem 198).
Hence, assumption \ref{a:concave} holds.
\end{proof}

Next, we aim to prove proposition \ref{pr:optpol_concave}. Recall $\cC_1$ given by \eqref{eq::cC1}. Consider a further subspace $\cC_2$ defined by
\begin{equation}
\label{eq:cC2}
  \cC_2 := \left\{ c \in \cC_1 \colon 
                   a \mapsto c(a,z) \text{ is concave for all }
                   z \in \ZZ 
           \right\}.
\end{equation}

\begin{lemma}
\label{lm:self_map_cC2}
$\cC_2$ is a closed subset of the metric space 
$(\cC, \rho)$, and $T c \in \cC_2$ for all $c \in \cC_2$.
\end{lemma}

\begin{proof}[Proof of lemma \ref{lm:self_map_cC2}]
The proof of the first claim is straightforward and thus omitted. 
We now prove the second claim.
Fix $c \in \cC_2$. By lemma \ref{lm:cC1} we have $Tc \in \cC_1$. It 
remains to show that $a \mapsto \xi(a, z) := Tc (a,z)$ is concave for all
$z \in \ZZ$. Given $z \in \ZZ$, lemma \ref{lm:binding} implies that $\xi(a,z) = a$ for $a \leq \bar{a}_c(z)$ and that $\xi (a,z) < a$ for $a > \bar{a}_c(z)$.
Since in addition $a \mapsto \xi(a,z)$ is continuous and increasing, to show the concavity of 
$\xi$ with respect to $a$, it suffices to show that $a \mapsto \xi (a,z)$ is 
concave on $(\bar{a}_c(z), \infty)$.

Suppose to the contrary that there exist some $z \in \ZZ$, 
$\alpha \in [0,1]$, and $a_1, a_2 \in (\bar{a}_c (z), \infty)$ such that 
\begin{equation}
\label{eq:as_ctdt}
  \xi \left( \alpha a_1 + (1 - \alpha) a_2, \, z \right)
  < \alpha \xi(a_1, z) + (1 - \alpha) \xi(a_2, z).
\end{equation}
Let $h(a, z, \hat{\omega}):= \hat{R} \left[a - \xi(a, z) \right] + \hat{Y}$, where $\hat{\omega} := (\hat{R}, \hat{Y})$. Then by lemma \ref{lm:binding} (and the analysis that follows immediately after that lemma), we have
\begin{align*}
(u' \circ \xi) \left( \alpha a_1 + (1 - \alpha) a_2, \, z \right)  
&=
  \beta \EE_z \hat{R} 
       \left(u' \circ c \right) 
           \left\{
                 h [\alpha a_1 + (1 - \alpha) a_2, \, z, \, \hat{\omega}],
                 \, \hat{z}
           \right\}    \\
&\leq 
  \beta \EE_z \hat{R}
       \left( u' \circ c \right) 
           \left[ 
                \alpha h(a_1, z, \hat{\omega}) + 
                (1 - \alpha) h(a_2, z, \hat{\omega}),
                \, \hat{z}
           \right].
\end{align*}
Using assumption \ref{a:concave} then yields 
\begin{align*}
\xi (\alpha a_1 + (1 - \alpha) a_2,  z)  
& \geq 
   (u')^{-1} \left\{
                \beta \EE_z \hat{R}
                         \left(u' \circ c \right) 
                         \left[
                              \alpha h(a_1, z, \hat{\omega}) + 
                              (1 - \alpha) h(a_2, z, \hat{\omega}),
                              \, \hat{z}
                         \right]                           
             \right\}    \\
& \geq 
   \alpha \left(u' \right)^{-1} 
          \left\{
                \beta \EE_z \hat{R}
                         \left(u' \circ c \right) 
                         \left[
                              h(a_1, z, \hat{\omega}), \, \hat{z}
                         \right]                           
             \right\}  +     \\ 
   & \quad \;  
   (1 -\alpha) \left(u' \right)^{-1} 
             \left\{
                \beta \EE_z \hat{R}
                         \left(u' \circ c \right) 
                         \left[
                              h(a_2, z, \hat{\omega}), \, \hat{z}
                         \right]                           
             \right\}    \\ 
& = \alpha \left(u' \right)^{-1} 
               \left\{ 
                    \left( u' \circ \xi \right) (a_1, z) 
               \right\} +
    (1 - \alpha) \left(u' \right)^{-1} 
                     \left\{ 
                         \left(u' \circ \xi \right) (a_2, z) 
                     \right\}   \\
& = \alpha \, \xi(a_1, z) + (1 - \alpha) \, \xi(a_2, z).
\end{align*}
This contradicts our assumption in \eqref{eq:as_ctdt}. Hence, 
$a \mapsto \xi(a,z)$ is concave for all $z \in \ZZ$. This concludes the
proof.
\end{proof}

Now we are ready to prove proposition \ref{pr:optpol_concave}.

\begin{proof}[Proof of proposition \ref{pr:optpol_concave}]
By theorem \ref{t:ctra_T}, we know that $T \colon \cC \rightarrow \cC$ is a 
contraction mapping with unique fixed point $c^*$. Since $\cC_2$ is a closed
subset of $\cC$ and $T: \cC_2 \rightarrow \cC_2$ by lemma \ref{lm:self_map_cC2}, we know that 
$c^* \in \cC_2$. The first claim is verified.

Regarding the second claim, note that $c^* \in \cC_2$ implies that 
$a \mapsto c^*(a,z)$ is increasing and concave for all $z \in \ZZ$. Hence, 
$a \mapsto \frac{c^*(a,z)}{a}$ is a decreasing function for all $z \in \ZZ$. 
Since in addition $c^*(a,z) \geq \alpha a$ for all $(a,z) \in \SS_0$ by 
proposition \ref{pr:opt_pol_bd_frac}, we know that
$\alpha' := \lim_{a \rightarrow \infty} \frac{c^*(a,z)}{a}$ is well-defined and 
$\alpha' \geq \alpha$. Finally, $\alpha' < 1$ by lemma \ref{lm:binding} and the fact that $\bar{a}(z) < \infty$ (see footnote \ref{fn:abar<inf}). Hence, the second claim holds.
\end{proof}

\subsection{Proof of section \ref{ss:glb_stb} results.}

We first prove the general result that the borrowing constraint binds in finite 
time with positive probability.

\begin{lemma}
	\label{lm:bind_fntime}
	For all $(a,z) \in \SS$, we have
	$\PP_{a,z} \left( \cup_{t \geq 0} \{ c_t = a_t \} \right) > 0$.
\end{lemma}

\begin{proof}[Proof of lemma \ref{lm:bind_fntime}]

The claim holds trivially when $a=0$. Suppose the claim does not hold on $\SS_0$ (recall that $\SS_0 = \SS \backslash \{0\}$), then $\PP_{a,z} \left( \cap_{t \geq 0} \{c_t < a_t\} \right) = 1$ for some $(a,z) \in \SS_0$, i.e., the borrowing constraint never binds with probability one. Hence, 
\begin{equation*}
  \PP_{a,z} 
     \left\{ 
        (u' \circ c)(a_t, z_t) = 
        \beta \EE \left[ 
                     R_{t+1} 
                     (u' \circ c) (a_{t+1}, z_{t+1}) \big| \fF_{t} 
                  \right]
     \right\}
  = 1
\end{equation*}
for all $t \geq 0$, where $\fF_t := \sigma(s_0, \cdots, s_t)$ with
$s_t := (a_t, z_t)$. Then we have
\begin{align}
\label{eq:u'c_ineq}
  \left( u' \circ c \right)(a,z)
  &= \beta^t \EE_{a,z} \,
                R_1 \cdots R_t 
                \left(u' \circ c \right)(a_t, z_t)    \nonumber  \\
  & \leq \beta^t \EE_{a,z} \,
                    R_1 \cdots R_t  
                    \left[ u'(a_t) + K \right]   \nonumber  \\
  & \leq \beta^t \EE_{z} \,
                     R_1 \cdots R_t 
                     \left[ u'(Y_t) + K \right]
\end{align}
for all $t \geq 1$. Let $t= kn +1$, where $n$ is the integer defined by assumption \ref{a:ctra_coef}. Based on assumption \ref{a:bd_sup_ereuprm} and the Markov property,
\begin{align*}
\beta^t \EE_z R_1 \cdots R_t 
&= \beta^t \EE_z R_1 \cdots R_{t-1} \EE_z (R_t \mid \fF_{t-1})
    = \beta^{t-1} \EE_z R_1 \cdots R_{t-1} \beta \EE_{z_{t-1}} R_1    \\
&\leq \left(\beta \sup_{z \in \ZZ} \EE_z R_1 \right) 
    (\beta^{nk} \EE_z R_1 \cdots R_{nk})
    \leq \left(\beta \sup_{z \in \ZZ} \EE_z R_1 \right) \theta^k
\to 0
\end{align*}
as $t \to \infty$, where $\theta \in [0,1)$ is given by assumption \ref{a:ctra_coef}. Similarly,
\begin{align*}
\beta^t \EE_z R_1 \cdots R_t u'(Y_t)
&= \beta^t \EE_z R_1 \cdots R_{t-1} 
          \EE_z \left[ R_t u'(Y_t) \mid \fF_{t-1} \right]     \\
&\leq \beta^{t} \EE_z R_1 \cdots R_{t-1} \EE_{z_{t-1}} 
    \left[R_1 u'(Y_1) \right]     \\
&\leq \left( \beta \sup_{z \in \ZZ} \EE_z [\hat{R} u'(\hat{Y})] \right) 
    \beta^{nk} \EE_z R_1 \cdots R_{nk}    \\
&\leq \left( \beta \sup_{z \in \ZZ} \EE_z [\hat{R} u'(\hat{Y})] \right) \theta^k
\to 0
\end{align*}
as $t \to \infty$. Letting $t \to \infty$. \eqref{eq:u'c_ineq} implies that
$\left(u' \circ c \right)(a,z) \leq 0$, contradicted with the fact that $u'>0$. 
Thus, we must have $\PP_{a,z} \left( \cup_{t \geq 0} \{ c_t = a_t\} \right) > 0$ 
for all $(a,z) \in \SS$. 
\end{proof}

\subsubsection{Proof of section \ref{ss:gs_iid} results}
The next few results establish global stability and the law of large numbers for the case of {\sc iid} $\{ z_t \}$ process.

We say that a stochastic kernel $Q$ \emph{increasing} if $s \mapsto \int h(s') Q(s, \diff s')$ is bounded and increasing whenever $h \colon \SS \to \RR$ is.

\begin{proof}[Proof of theorem \ref{t:gs_iid}]
Obviously, assumptions \ref{a:utility}, \ref{a:ctra_coef}--\ref{a:conti_ereuprm} and \ref{a:suff_bd_in_prob}--\ref{a:concave} hold under the stated assumptions of theorem \ref{t:gs_iid}. Based on proposition \ref{pr:optpol_concave}, we have $c^* \in \cC_2$. In particular, $a \mapsto c^*(a)$ is continuous, and $a \mapsto \frac{c^*(a)}{a}$ is decreasing on $(0, \infty)$. Hence, $a_{t+1}$ is continuous and increasing in $a_t$ (see equation \eqref{eq:dyn_sys_iid}). The stochastic kernel $Q$ is then Feller and increasing. Moreover, $Q$ is bounded in probability by lemma \ref{lm:bd_in_prob_at}. 

Fix $a_0$ and $a_0'$ in $\RR_+$ with $a_0' \leq a_0$. Let $\{a_t\}$ and $\{a_t'\}$ be two independent Markov processes generated by \eqref{eq:dyn_sys_iid}, starting at $a_0$ and $a_0'$ respectively. Let $\{ c_t\}$ and $\{c_t'\}$ be the corresponding optimal consumption paths. By lemma \ref{lm:bind_fntime}, $\PP_{a_0} (\cup_{t \geq 0} \{c_t = a_t \}) > 0$, i.e., the borrowing constraint binds in finite time with positive probability. Hence, with positive probability, $a_{t+1} = Y_{t+1} \leq R_{t+1} (a_t' - c_t') + Y_{t+1} = a_{t+1}'$. In other words, 
$\PP \{ a_{t+1} \leq a_{t+1}'\} > 0$ and $Q$ is order reversing. 

Since $Q$ is increasing, Feller, order reversing, and bounded in probability, based on theorem~3.2 of \cite{kamihigashi2014stochastic}, $Q$ is globally stable.
\end{proof}

\begin{proof}[Proof of theorem \ref{t:LLN_iid}]
We have shown in the proof of theorem \ref{t:gs_iid} that the stochastic kernel $Q$ is increasing, bounded in probability, and order reversing. Hence, $Q$ is monotone ergodic by proposition~4.1 of \cite{kamihigashi2016seeking}. The two claims of theorem \ref{t:LLN_iid} then follow from theorem \ref{t:gs_iid} (of this paper), and corollary~3.1 and theorem~3.2 of \cite{kamihigashi2016seeking}. In particular, if we pair $\SS$ with its usual pointwise order $\leq$, then assumption~3.1 of \cite{kamihigashi2016seeking} obviously holds.
\end{proof}

\subsubsection{Proof of Section \ref{ss:gs_general} Results}

Our next goal is to prove theorems \ref{t:gs_gnl_ergo_LLN}--\ref{t:gs_gnl}. In proofs we apply the theory of \cite{meyn2009markov}. Important definitions  (their locations in \cite{meyn2009markov}) include: $\psi$-irreducibility (section~4.2), small set (page~102), strong aperiodicity (page~114), petite set (page~117), Harris chain (page~199), and positivity (page~230). 
  
Note that since $\RR^m$ paired with its Euclidean topology is a second countable 
topological space (i.e., its topology has a countable base), while $\RR_+$ and 
$\ZZ$ are respectively Borel subsets of $\RR$ and $\RR^m$ paired with the 
relative topologies, $\RR_+$ and $\ZZ$ are also second countable. As a result, 
for $\SS := \RR_+ \times \ZZ$, it always holds that (see, e.g., page~149, 
theorem~4.44 of \cite{guide2006infinite})
\begin{equation*}
    \bB(\SS) = \bB (\RR_+) \otimes \bB(\ZZ).
\end{equation*}
Recall the Lebesgue measure $\lambda$ on $\bB(\RR_+)$ and the measure $\vartheta$ 
on $\bB(\ZZ)$ defined in section \ref{ss:gs_general}. Let 
$\lambda \times \vartheta$ be the product measure on $\bB (\SS)$.

%With slight abuse of notation, throughout this section, we use $\bar{a}$ to 
%denote the threshold function (recall \eqref{eq:a_bar}) related to the optimal 
%policy $c^*$, in other words, 
%%
%\begin{equation*}
%    \bar{a}(z) := (u')^{-1} 
%    \left[ 
%        \beta \EE_z \hat{R} (u' \circ c^*) (\hat{Y}, \hat{z}) 
%    \right].
%\end{equation*}

\begin{lemma}
	\label{lm:inf_abar}
	Let the function $\bar{a}$ be defined as in \eqref{eq:a_bar}. Then $\inf_{z \in \ZZ} \bar{a}(z) > 0$.
\end{lemma}

\begin{proof}[Proof of lemma \ref{lm:inf_abar}]
	Since $c^* \in \cC$, there exists a constant $K > 0$ such that 
	\begin{equation*}
	    0 < (u' \circ c^*) (a, z) \leq u'(a) + K
	    \quad \text{for all } (a,z) \in \SS_0.
	\end{equation*}
	Assumption \ref{a:bd_sup_ereuprm} then implies that
	\begin{equation*}
	    \sup_{z \in \ZZ} \EE_z \hat{R} (u' \circ c^*) (\hat{Y}, \hat{z})
	    \leq \sup_{z \in \ZZ} \EE_z \hat{R} u' (\hat{Y}) + 
	    K \sup_{z \in \ZZ} \EE_z \hat{R} < \infty.
	\end{equation*} 
	Then, by the definition of $\bar{a}$ and the properties of $u$, 
	\begin{equation*}
		\inf_{z \in \ZZ} \bar{a} (z) 
%		= \inf_{z \in \ZZ} \, (u')^{-1} 
%		\left[
%		    \beta \EE_z \hat{R} (u' \circ c^*) (\hat{Y}, \hat{z})
%		\right]
		= (u')^{-1} 
		\left[
		    \beta \sup_{z \in \ZZ} \EE_z \hat{R} (u' \circ c^*) (\hat{Y}, \hat{z})
		\right]
		> 0,
	\end{equation*}
	as claimed.
\end{proof}

Recall the compact subset $\CC \subset \ZZ$ and $\delta_Y > 0$ given by assumption \ref{a:pos_dens}. Let
\begin{equation}
	\label{eq:C'D}
	\CC' := \left[
	0, \, \min \left\{ \delta_Y, \, \inf_{z \in \ZZ} \bar{a}(z) \right\} 
	\right]
	\quad \text{and} \quad
	\DD := \CC' \times \CC \in \bB(\SS).
\end{equation}

\begin{lemma}
	\label{lm:psi_irr}
	The Markov process $\{ (a_t,z_t) \}_{t \geq 0}$ is $\psi$-irreducible.
\end{lemma}

\begin{proof}[Proof of lemma \ref{lm:psi_irr}]

	We define the measure $\varphi$ on $\bB(\SS)$ by
	\begin{equation*}
	    \varphi(A) := (\lambda \times \vartheta) (A \cap \DD) 
	    \quad \text{for } A \in \bB (\SS).
	\end{equation*}
	Then $\varphi$ is a nontrivial measure. In particular,
	$\varphi (\SS) = (\lambda \times \vartheta) (\DD) = \lambda(\CC') 
	\vartheta(\CC) > 0$ since $\lambda (\CC') > 0$ by lemma \ref{lm:inf_abar} and 
	$\vartheta (\CC) > 0$ by assumption \ref{a:pos_dens}.
	
	For fixed $(a,z) \in \SS$ and $A \in \bB(\SS)$ with $\varphi (A) > 0$, by lemma
	\ref{lm:binding},
	\begin{align}
	    \label{eq:bind_ineq}
	    \PP_{(a,z)} \{ (a_{t+1}, z_{t+1}) \in A \} 
	    &\geq \PP_{(a,z)} \{ (a_{t+1}, z_{t+1}) \in A, \, a_t \leq \bar{a}(z_t) \}    \nonumber \\
	    &= \PP_{(a,z)} \{ (a_{t+1}, z_{t+1}) \in A \mid c_t = a_t \} \,
	    \PP_{(a,z)} \{ c_t = a_t \}    \nonumber \\
	    &=\PP_{(a,z)} \{ (Y_{t+1}, z_{t+1}) \in A \mid c_t = a_t \} \,
	    \PP_{(a,z)} \{ c_t = a_t \}    \nonumber \\
	    &=\PP_{(a,z)} \{ (Y_{t+1}, z_{t+1}) \in A, \, a_t \leq \bar{a}(z_t) \}.
	\end{align}
	Note that for all $z' \in \ZZ$, by assumption \ref{a:pos_dens}, 
	$f_L(Y'' \mid z'') p(z'' \mid z') > 0$ whenever $(Y'', z'') \in \DD$. Since in 
	addition $\varphi (A) = (\lambda \times \vartheta)(A \cap \DD) > 0$, we have
	\begin{equation*}
	    \int_A f_L(Y'' \mid z'') p(z'' \mid z') 
	    (\lambda \times \vartheta) [\diff (Y'', z'')] > 0
	    \quad \text{for all } z' \in \ZZ.
	\end{equation*}
	Let $\triangle := \PP_{(a,z)} \{ (a_{t+1}, z_{t+1}) \in A \}$ and
	$E := \{ (a',z') \in \SS \colon a' \leq \bar{a}(z') \}$. Notice that by lemma 
	\ref{lm:binding} and lemma \ref{lm:bind_fntime}, there exists $t \in \NN$ such 
	that 
	\begin{equation*}
	    Q^t \left( (a,z), E \right) = \PP_{(a,z)} \{ a_t \leq \bar{a}(z_t) \} > 0.
	\end{equation*}
	Hence, \eqref{eq:bind_ineq} implies that
	\begin{align*}
	    \triangle 
	    &\geq \int_E 
	        \left\{
	            \int_A f_L (Y'' \mid z'') p(z'' \mid z') 
	            (\lambda \times \vartheta) [\diff (Y'', z'')]
	        \right\}
	    Q^t \left( (a,z), \diff (a', z') \right) > 0.
	\end{align*}
	Therefore, we have shown that any measurable subset with positive $\varphi$ 
	measure can be reached in finite time with positive probability, i.e., 
	$\{ (a_t,z_t)\}$ is $\varphi$-irreducible. Based on proposition~4.2.2 of 
	\cite{meyn2009markov}, there exists a maximal (in the sense of absolute 
	continuity) probability measure $\psi$ on $\bB(\SS)$ such that 
	$\{ (a_t,z_t)\}$ is $\psi$-irreducible.
\end{proof}

\begin{lemma}
	\label{lm:str_aperi}
	The Markov process $\{ (a_t, z_t)\}_{t \geq 0}$ is strongly aperiodic.
\end{lemma}
	
\begin{proof}[Proof of lemma \ref{lm:str_aperi}]
	By the definition of strong aperiodicity, we need to show that there exists 
	a $v_1$-small set $D$ with $v_1 (D) > 0$, i.e., there exists a nontrivial 
	measure $v_1$ on $\bB (\SS)$ and a subset $D \in \bB (\SS)$ such that 
	$v_1 (D) > 0$ and
	\begin{equation}
		\label{eq:small}
		\inf_{(a,z) \in D}  
		Q \left((a,z), A \right) \geq v_1 \left( A \right)
		\quad \text{for all }
		A \in \bB (\SS).
	\end{equation}
	Let $\DD$ be defined as in \eqref{eq:C'D}. We show that $\DD$ satisfies the 
	above conditions. Let
	\begin{equation*}
		r(a', z') := f_L (a' \mid z') \, \inf_{z \in \CC} p(z' \mid z),
		\qquad (a',z') \in \SS.
	\end{equation*}
	Since by assumption \ref{a:pos_dens}, $p(z' \mid z)$ is strictly positive on
	$\CC \times \ZZ$ and continuous in $z$, and $f_L(Y' \mid z')$ is strictly 
	positive on $(0, \delta_Y) \times \CC$, the definition of $\DD$ implies that 
	$r(a',z')$ is strictly positive whenever $(a',z') \in \DD$. Define the measure 
	$v_1$ on $\bB(\SS)$ by 
	\begin{equation*}
	    v_1 (A) := \int_A r(a',z') (\lambda \times \vartheta) [\diff (a',z')]
	    \quad \text{for } A \in \bB(\SS).
	\end{equation*}
	Since $(\lambda \times \vartheta) (\DD) > 0$ as shown in the proof of lemma 
	\ref{lm:psi_irr} and $r(a',z') >0$ on $\DD$, we have $v_1 (\DD) > 0$, which 
	also implies that $v_1$ is a nontrivial measure.
	
	Let $g[(a',z') \mid (a,z)]$ denote the density representation of the stochastic 
	kernel $Q$ when $(a,z) \in \DD$. Lemma \ref{lm:binding} implies that
	\begin{equation*}
		g[(a',z') \mid (a,z)] = f_L (a' \mid z') p(z'\mid z),
		\qquad (a,z) \in \DD.
	\end{equation*}
	Hence, for all $(a,z) \in \DD$ and $A \in \bB(S)$, 
	\begin{align*}
	    Q \left((a,z), A \right) 
	    &= \int_A g[(a',z') \mid (a,z)] 
	    (\lambda \times \vartheta) [\diff (a',z')]    \\
	    &\geq \int_A r(a',z') (\lambda \times \vartheta) [\diff (a',z')]
	    = v_1 (A). 
	\end{align*}
	This implies that condition \eqref{eq:small} holds. Hence, 
	$\{(a_t, z_t)\}_{t \geq 0}$ is strongly aperiodic.
\end{proof}

%A point $(a^*, z^*) \in \SS$ is called \emph{reachable} if for every open set
%$O \in \bB(\SS)$ containing $(a^*, z^*)$, we have 
%$\sum_{n \geq 0} Q^n \left( \left(a,z \right), O \right) > 0$ for all 
%$(a,z) \in \SS$.

\begin{proof}[Proof of theorem \ref{t:gs_gnl_ergo_LLN}]
	
	We first show that $\{(a_t, z_t)\}$ is a positive Harris chain. Positivity 
	has been established in theorem~\ref{t:sta_exist}. To show Harris recurrence, 
	by lemma~6.1.4, theorem~6.2.9 and theorem~18.3.2 of \cite{meyn2009markov}, it 
	suffices to verify 
	\begin{enumerate}
		\item[(a)] $Q$ is Feller and bounded in probability, and    
		\item[(b)] $\{ (a_t, z_t)\}$ is $\psi$-irreducible, and the support of $\psi$ has non-empty interior.
	\end{enumerate}
	Claim~(a) is already proved in theorem \ref{t:sta_exist}. Regarding 
	claim~(b), in lemma~\ref{lm:psi_irr} we have shown that $\{(a_t,z_t) \}$ is
	$\varphi$-irreducible and thus $\psi$-irreducible, where $\psi$ is maximal in 
	the sense that $\psi(A) = 0$ implies $\varphi (A)=0$ for all 
	$A \in \bB(\SS)$. This also implies that $\psi (A) > 0$ whenever 
	$\varphi (A) > 0$. 
	Recall that $\varphi (A) := (\lambda \times \vartheta) (A \cap \DD)$, where 
	$\DD := \CC' \times \CC$ is defined by \eqref{eq:C'D}. Since by assumption
	\ref{a:pos_dens}, the support of $\vartheta$ contains $\CC$ that has nonempty 
	interior and the support of $\lambda$ (the Lebesgue measure) contains the
	interval $\CC'$ (of positive $\lambda$ measure), the support of $\varphi$ 
	contains $\DD = \CC' \times \CC$ that has nonempty interior. As a result, the 
	support of $\psi$ contains $\DD$ and thus has nonempty interior. Claim~(b) is 
	verified. Therefore, $\{(a_t, z_t)\}$ is a positive Harris chain.
	
	Since in addition we have shown in lemmas 
	\ref{lm:psi_irr}--\ref{lm:str_aperi} that $\{(a_t, z_t)\}$ is 
	$\psi$-irreducible and strongly aperiodic, based on theorem~13.0.1 and 
	theorem~17.1.7 of \cite{meyn2009markov}, the stated claims of our theorem hold. This concludes the proof.
\end{proof}

Our next goal is to prove theorem \ref{t:gs_gnl}. We start by proving several lemmas.

\begin{lemma}
	\label{lm:v2small}
	The set $B:= [0, d] \times \{z\}$ is a petite set for all $d \in (0, \infty)$ and $z \in \ZZ$.  
\end{lemma}

\begin{proof}[Proof of lemma \ref{lm:v2small}]
	Since any small set is petite, it suffices to show that $B$ is a $v_2$-small 
	set, i.e., there exists a nontrivial measure $v_{2}$ on $\bB(\SS)$ such that
	\begin{equation}
	\label{eq:v2small}
	  \inf_{(a,z) \in B} Q^2((a,z), A) \geq v_{2} (A) 
	  \quad \text{for all }  A \in \bB(\SS).
	\end{equation}
	Without loss of generality, we assume that $d$ is large enough.
	For $a \neq c^*(a,z)$, let 
	\begin{equation}
		\label{eq:f_dens}
		f \left(a' \mid a,z,z' \right) 
		:= \frac{1}{a - c^*(a,z)} 
		\int_{[0,a']}  
		    f_C \left( \frac{a' - Y'}{a - c^*(a,z)} \, \Big| \, z' \right)
		    f_L \left( Y' \mid z' \right)
		\diff Y',
	\end{equation}
	while $f \left( \cdot \mid a,z,z' \right) := f_L (\cdot \mid z')$ for 
	$a = c^*(a,z)$. Let $g \left[ (a',z') \mid (a,z) \right]$ be the density 
	corresponding to the stochastic kernel $Q$. Since $\{ \zeta_t\}$ and 
	$\{ \eta_t\}$ are mutually independent by assumption \ref{a:geo_drift_Yt}, 
	$g$ satisfies
	\begin{equation*}
		\label{eq:g_dens}
		g \left[ \left( a',z' \right) \mid (a,z) \right]
		= f \left( a' \mid a,z, z' \right) p \left( z' \mid z \right).
	\end{equation*}
	Recall that we have shown in the proof of proposition \ref{pr:optpol_concave} 
	that $a \mapsto c^*(a,z) / a$ is decreasing for all $z \in \ZZ$. This implies that,
	for the dynamical system \eqref{eq:dyn_sys}, $a_{t+1}$ is increasing in $a_t$ 
	with probability one.
	Since in addition $c^*(a,z) = a$ if and only if $a \leq \bar{a}(z)$ by lemma 
	\ref{lm:binding}, we have
	\begin{align*}
	Q^2((a,z), A) 
	  &= \PP_{a,z} \left\{ (a_2, z_2) \in A \right\}
	    \geq \PP_{a,z} \left\{ (a_2, z_2) \in A, \, a_1 \leq \bar{a}(z_1) \right\}    \\
	  &= \PP_{a,z} \left\{ (a_2, z_2) \in A \mid  a_1 \leq \bar{a}(z_1) \right\}    
	        \PP_{a,z} \left\{ a_1 \leq \bar{a}(z_1) \right\}    \\
	  &= \PP_{a,z} \left\{ (Y_2, z_2) \in A 
	                                \mid 
	                                a_1 \leq \bar{a}(z_1)
	                      \right\}    
	        \PP_{a,z} \left\{ a_1 \leq \bar{a}(z_1) \right\}    \\
	  &= \PP \left\{ (Y_2, z_2) \in A, \,
	                         a_1 \leq \bar{a}(z_1)
	                         \mid
	                         (a_0,z_0) = (a,z)
	              \right\}    \\
	  &\geq \PP \left\{ (Y_2, z_2) \in A, \,
	                              a_1 \leq \bar{a}(z_1)
	                              \mid
	                              (a_0,z_0) = (d,z)
	                   \right\} =: v_2(A)
	\end{align*}
	for all $(a,z) \in B$, where the last inequality follows from the fact that 
	$a_{t+1}$ is increasing in $a_t$ (shown above), which indicates that 
	for all fixed $(a, z) \in B$ and $z_1 \in \ZZ$, 
	\begin{equation*}
	  \int f(a_1\mid a, z, z_1) \1 \{ a_1 \leq \bar{a}(z_1) \} \diff a_1
	  \geq \int f(a_1\mid d, z, z_1) \1 \{ a_1 \leq \bar{a}(z_1) \} \diff a_1 > 0.
	\end{equation*}
	We now show that $v_2$ defined this way is a nontrivial measure on 
	$\bB(\SS)$. Obviously, $v_2$ is a measure. Moreover, for fixed $z \in \ZZ$,
	$c^*(a,z) / a$ is decreasing in $a$, strictly less than one as $a$ gets 
	large, and bounded below by $\alpha \in (0,1)$. Hence, there exists 
	$\alpha' \in (0,1)$ such that $c^*(a,z) / a \leq \alpha'$ as $a$ gets large. 
	Hence, $a - c^*(a,z) \geq (1 - \alpha') a$, which implies that
	$a - c^*(a,z) \to \infty$ as $a \to \infty$. Using lemma \ref{lm:binding} 
	again shows that $f(a' \mid a,z,z')$ satisfies \eqref{eq:f_dens} as $a$ gets 
	large. Let 
	$\underline{a} := \inf_{z \in \ZZ} \bar{a}(z)$. Then $\underline{a} > 0$ by 
	lemma \ref{lm:inf_abar}. Recall $\delta_R > 0$, $\delta_Y > 0$ and the 
	compact subset $\CC \subset \ZZ$ defined by assumption \ref{a:pos_dens}. Then
	\begin{equation*}
	    0 < \frac{\underline{a}}{d - c^*(d,z)} < \delta_R
	    \quad \text{as $d$ gets large.}
	\end{equation*}
	Since in addition $f_L (Y \mid z)$ is strictly positive on 
	$(0, \delta_Y) \times \CC$ and $f_C (R \mid z)$ is strictly positive on 
	$(0, \delta_R) \times \CC$ by assumptions 
	\ref{a:pos_dens}--\ref{a:geo_drift_Yt}, for $d$ that is large enough, 
	$f(a' \mid d,z,z')$ is defined by \eqref{eq:f_dens} and it is strictly 
	positive for all $(a',z') \in (0, \underline{a}) \times \CC$. 
	Moreover, since $p(z' \mid z)$ is strictly positive on $\CC \times \ZZ$ and 
	$\vartheta (\CC) > 0$ by assumption \ref{a:pos_dens}, 
	\begin{align*}
	    v_2 (\SS) 
	    &= \PP_{(d,z)} \{ a_1 \leq \bar{a} (z_1) \} 
	    \geq \PP_{(d,z)} \left\{ a_1 \leq \underline{a} \right\}  \\
	    &= \int_{\ZZ} 
	        \left[
	            \int_{[0, \underline{a}]} f(a' \mid d, z, z') \diff a'
	        \right]
	        p(z'\mid z) 
	    \vartheta (\diff z') > 0.
	\end{align*}
	Hence, $v_2$ is a nontrivial measure on $\bB (\SS)$. Since in addition $z$ 
	is the only element of $\ZZ$ that appears in the analytical form of $B$, 
	\eqref{eq:v2small} holds and thus $B$ is petite. 
\end{proof}

In the following, we let $\alpha \in [0,1)$ and $n \in \NN$ be defined as in assumption \ref{a:suff_bd_in_prob}.

\begin{lemma}
	\label{lm:geo_drift}
	There exist a petite set $B$, constants $b < \infty$, $\rho > 0$ and a 
	measurable map $V \colon \SS \rightarrow [1, \infty)$ such 
	that, for all $(a,z) \in \SS$,
	\begin{equation*}
	  \EE_{a,z} V(a_n, z_n) - V(a,z) 
	  \leq - \rho V(a,z) + b \1 \{(a,z) \in B \}.
	\end{equation*}
\end{lemma}

\begin{proof}[Proof of lemma \ref{lm:geo_drift}]

	By assumption \ref{a:geo_drift_Yt}, there exists $q'' \in \RR_+$ such that
	\begin{equation*}
	    \EE_z Y_t \leq q^{t-1} \EE_z Y_1 + q''
	    \quad \text{for all } t \in \NN \text{ and } z \in \ZZ.
	\end{equation*}
	Since $c^*(a,z) \geq \alpha a$ for all $(a,z) \in \SS$ by proposition 
	\ref{pr:opt_pol_bd_frac}, $M:= \sup_{z \in \ZZ} \EE_z \hat{R} < \infty$ by 
	assumption \ref{a:bd_sup_ereuprm}, and
	$\gamma := (1 - \alpha)^n \sup_{z \in \ZZ} \EE_z R_n \cdots R_1 < 1$ 	
	by assumption \ref{a:suff_bd_in_prob}, we have
	\begin{align*}
	    \EE_{a,z} a_n 
	    &\leq (1 - \alpha)^n \EE_z R_n \cdots R_1 a + 
	    \sum_{t = 1}^{n} (1 - \alpha)^{n-t} \EE_z R_n \cdots R_{t+1} Y_t    \\
	    &\leq \gamma a + \sum_{t=1}^{n} (1 - \alpha)^{n-t} M^{n-t} \EE_z Y_t
	    \leq \gamma a + \sum_{t=1}^{n} (1 - \alpha)^{n-t} M^{n-t} (q^{t-1} \EE_z Y_1 + q'').
	\end{align*}
	Define $L := \sum_{t=1}^{n} (1 - \alpha)^{n-t} M^{n-t}$ and 
	$\tilde{L} := q'' L$. Then $L, \tilde{L} \in \RR_+$ and the above inequality 
	implies that
	\begin{align*}
	    \EE_{a,z} a_n \leq \gamma a + L \EE_z Y_1 + \tilde{L} 
	    \quad \text{for all } (a,z) \in \SS.
	\end{align*}
	Choose $m \in \RR_+$ such that $1 - q^n - L / m > 0$ (such an $m$ is 
	available since $q \in [0,1)$ by assumption \ref{a:geo_drift_Yt}). Let $V$ be 
	defined as in \eqref{eq:V_func}, i.e., $V(a,z) = a + m \EE_z Y_1 + 1$.
	
	Then the above results imply that
	\begin{align*}
		\EE_{a,z} V(a_n, z_n) 
		&= \EE_{a,z} a_n + m \, \EE_z \EE_{z_n} Y_1 + 1  
		= \EE_{a,z} a_n + m \, \EE_z Y_{n+1} + 1   \\
		& \leq \gamma a + L \EE_z Y_1 + \tilde{L} + m (q^n \EE_z Y_1 + q'') + 1    \\
		&= \gamma a + (L/m + q^n) m \EE_z Y_1 + \tilde{L} + m q'' + 1.
	\end{align*}
	Let $\tilde{\rho} 
	     := \min \left\{ 
	                 1 - \gamma, \, 1 - q^n - L/m
	             \right\}$.
	Then $\tilde{\rho} > 0$ by assumption \ref{a:suff_bd_in_prob} and the 
	construction of $m$. Thus, 
	\begin{align}
		\label{eq:dri_ineq1}
		&\EE_{a,z} V(a_n, z_n) - V(a, z)  \nonumber \\
		&\leq 
		   - (1 - \gamma) a
		   - \left( 1 - q^n - L / m \right) m \, \EE_z Y_1 
		   + \tilde{L} + m q''   \nonumber \\
		&\leq -\tilde{\rho} \left( a + m \, \EE_z Y_1 \right) + \tilde{L} + m q''
		  = -\tilde{\rho} V(a,z) + \tilde{\rho} + \tilde{L} + m q''.
	\end{align}
	Choose $\rho \in (0, \tilde{\rho})$ and $d \in \RR_+$ such that 
	$(\tilde{\rho} - \rho) d  > \tilde{\rho} + \tilde{L} + m q''$. 
	Fix $z_0 \in \ZZ$ and let $B:= [0,d] \times {z_0}$. Lemma \ref{lm:v2small} 
	implies that $B$ is a petite set. Notice that 
	\begin{equation*}
		V(a,z) = a + m \, \EE_z Y_1 + 1 > d
		\quad \text{ for all } (a,z) \notin B.
	\end{equation*}
	Hence, \eqref{eq:dri_ineq1} implies that for all $(a,z) \notin B$, we have
	\begin{align}
		\label{eq:dri_ineq2}
		&\EE_{a,z} V(a_1, z_1) - V(a, z)  
		  \leq -\tilde{\rho} V(a,z) + \tilde{\rho} + \tilde{L} + m q''  \nonumber \\
		&= -\rho V(a,z) - (\tilde{\rho} - \rho) V(a,z) 
		   + \tilde{\rho} + \tilde{L} + m q''    \nonumber \\ 
		&< -\rho V(a,z) - (\tilde{\rho} - \rho) d
		   + \tilde{\rho} + \tilde{L} + m q'' 
		 < -\rho V(a,z).
	\end{align}
	Let $b:= \tilde{\rho} + \tilde{L} + m q''$. Then by 
	\eqref{eq:dri_ineq1}--\eqref{eq:dri_ineq2}, we have
	\begin{equation*}
	  \EE_{a,z} V(a_n, z_n) - V(a, z)  
	  \leq -\rho V(a,z) + b \1 \{ (a,z) \in B \}
	\end{equation*}
	for all $(a,z) \in \SS$. This concludes the proof.
\end{proof}

\begin{proof}[Proof of theorem \ref{t:gs_gnl}]
	That $Q$ is $V$-geometrically ergodic can be proved by applying 
	theorem~19.1.3 (or proposition~5.4.5 and theorem~15.0.1) of 
	\cite{meyn2009markov}. All the required conditions in those theorems have 
	been established in our lemmas \ref{lm:psi_irr}--\ref{lm:geo_drift} above.
\end{proof}

\bibliographystyle{ecta}

\bibliography{localbib}

\end{document}